\documentclass[hyper]{prop2015}
\usepackage[english]{babel}
\usepackage[utf8]{inputenc}

\usepackage{amsmath, amsthm, amssymb, amsbsy, bm, mathrsfs}
\usepackage{enumerate}
\usepackage[toc,page]{appendix}
\usepackage{tikz-cd}
\usepackage{bbm}

\newtheorem{thm}{Theorem}[section]

\newtheorem{tvrz}[thm]{Proposition}

\newtheorem{theorem}[thm]{Theorem}
\theoremstyle{definition}

\theoremstyle{remark}
\newtheorem{rem}[thm]{Remark}
\theoremstyle{definition}
\newtheorem{example}[thm]{Example}

\def\<{\langle}
\def\>{\rangle}
\def\~{\widetilde}
\def\^{\wedge}

\def\tr{\triangleright}
\def\cD{\nabla}

\def\fb{\mathbf{b}}
\def\fk{\mathbf{k}}
\def\fa{\mathbf{a}}
\def\fc{\mathbf{c}}
\def\fh{\mathbf{h}}
\def\ftheta{\bm{\theta}}
\def\falpha{\bm{\alpha}}
\def\fA{\mathbf{A}}

\def\fE{\mathbf{E}}

\def\fM{\mathbf{M}}
\def\fg{\mathbf{g}}

\def\fB{\mathbf{B}}

\def\fK{\mathbf{K}}

\def\fPi{\bm{\Pi}}

\def\f1{\mathbf{1}}

\def\fPsi{\mathbf{\Psi}}

\def\gm{\mathbf{G}}

\def\g{\mathfrak{g}}

\def\d{\mathfrak{d}}

\def\su{\mathfrak{su}}

\def\X{\mathfrak{X}}

\def\frc{\mathfrak{c}}

\def\ti{\text{i}}
\def\tj{\text{j}}
\def\tq{\text{q}}
\def\io{\mathit{i}}

\def\E{\mathcal{E}}

\def\S{\mathcal{S}}
\def\RS{\mathcal{R}}

\def\R{\mathbbm{R}}

\def\gTM{\mathbbm{T}M}
\def\gTS{\mathbbm{T}S}

\newcommand{\bma}[4]{\begin{pmatrix} #1 & #2 \\ #3 & #4 \end{pmatrix}}

\newcommand{\Li}[1]{ \mathcal{L}_{#1}}

\DeclareMathOperator{\vol}{vol}
\DeclareMathOperator{\Ad}{Ad}
\DeclareMathOperator{\LC}{LC}
\DeclareMathOperator{\ad}{ad}
\DeclareMathOperator{\End}{End}

\DeclareMathOperator{\Hom}{Hom}

\DeclareMathOperator{\rk}{rk}
\DeclareMathOperator{\Aut}{Aut}

\DeclareMathOperator{\Sym}{Sym}
\DeclareMathOperator{\GL}{GL}
\DeclareMathOperator{\Ric}{Ric}

\DeclareMathOperator{\Lie}{Lie}

\DeclareMathOperator{\cyc}{cyclic}

\DeclareMathOperator{\Tr}{Tr}

\DeclareMathOperator{\Div}{div}

\category{Proceedings}
\keywords{(Quasi-)Poisson--Lie T-duality, effective actions, sigma models, Courant algebroids, generalized metric, dilaton, Levi-Civita connections}
\subtitle{\href{http://www.maths.dur.ac.uk/lms/109/index.html}{LMS/EPSRC Durham Symposium on Higher Structures in M-Theory}}
\title{Effective Actions for $\sigma$-Models of Poisson--Lie Type}
\author[B. Jurčo]{Branislav Jurčo\inst{a}}
\author[J. Vysoký]{Jan Vysoký\inst{b,}\footnote{Corresponding author e-mail:~\href{mailto:jan.vysoky@fjfi.cvut.cz}{\textsf{jan.vysoky@fjfi.cvut.cz}}}}
\address[1]{Mathematical Institute, Faculty of Mathematics and Physics, Charles University, Sokolovská 83, 18675 Prague 8, Czech Republic}
\address[2]{Faculty of Nuclear Sciences and Physical Engineering, Czech Technical University in Prague, Břehová 7, 11519 Prague 1, Czech Republic}
\begin{acknowledgements}
We would like to thank Pavol Ševera and Fridrich Valach for helpful discussions. The research of B.J. was supported by grant GAČR P201/12/G028. J.V. is grateful for a financial support from MŠMT under grant no. RVO 14000. He would also like to acknowledge the contribution of the COST Action MP1405.
\end{acknowledgements}

\begin{abstract}
(Quasi-)Poisson--Lie T-duality of string effective actions is described in the framework of generalized geometry of Courant algebroids. The approach is based on a generalization of Riemannian geometry in the context of Courant algebroids, including a proper version of a Levi-Civita connection. In our approach, the dilaton field is encoded in a Levi-Civita connection and its form is determined by the Courant algebroid geometry. Explicit examples of background solutions are provided using the approach developed in the paper.
\end{abstract}
\shortabstract
\begin{document}
\maketitle

\section{Introduction, Poisson--Lie T-duality}
In their original series of papers \cite{Klimcik:1995ux, Klimcik:1995dy, Klimcik:1995jn}, Klimčík and Ševera proposed a new kind of non-Abelian duality, the so called Poisson--Lie T-duality, between two-dimensional $\sigma$-models. Recently, this observation was interpreted in terms of geometry of Courant algebroids in \cite{Severa:2015hta,Severa:2016prq,Severa:2017kcs}. 

Let us recall the main statement of the Poisson--Lie T-duality in the case relevant for this paper. Details can be found in   \cite{Severa:2016prq, Severa:2017kcs} and, to some extent, in the following sections. A \textbf{$2$-dimensional $\sigma$-model} is a field theory given by the action functional
\begin{equation} \label{eq_SMaction}
\S_{\sigma}[\ell] = \int_{\Sigma} \<h, \ell^{\ast}(g)\>_{h} \cdot d\vol_{h} + \int_{\Sigma} \ell^{\ast}(B) + \int_{X} \ell^{\ast}(H),
\end{equation}
where the fields are smooth maps $\ell: \Sigma \rightarrow M$, and 
\begin{enumerate}[i)]
\item $(\Sigma,h)$ is a $2$-dimensional oriented pseudo-Rie- \linebreak mannian smooth manifold, the \textbf{worldsheet};
\item $M$ is a smooth manifold, the \textbf{target}, equipped with a metric $g$, a $2$-form $B$ and a closed $3$-form $H$;
\item $\Sigma$ is the boundary of a $3$-dimensional smooth manifold $X$ and $\ell$ in the last term is an arbitrary extension of the map $\ell: \Sigma \rightarrow M$ to $X$. 
\end{enumerate}
Using the Stokes theorem, the theory is manifestly invariant under the change $B \mapsto B - C$ and $H \mapsto H + dC$, where $C \in \Omega^{2}(M)$ is an arbitrary $2$-form. For this reason, the term with $B$ can be omitted without the loss of generality.

Now, for Poisson--Lie T-duality (in its simplified form), one assumes the target manifold to be a left coset space $D/G$, where we consider
\begin{enumerate}[i)]
\item $D$ being a connected Lie group whose Lie algebra $\d = \Lie(D)$ is equipped with a non-degenerate symmetric bilinear form $\<\cdot,\cdot\>_{\d}$;
\item $G \subset D$ being a connected closed Lie subgroup whose Lie algebra $\g = \Lie(G)$ is Lagrangian with respect to $\<\cdot,\cdot\>_{\d}$, that is $\g = \g^{\perp}$. 
\end{enumerate}
In accordance with \cite{Severa:2017kcs}, in this case, the action (\ref{eq_SMaction}) describes a \textbf{$\sigma$-model of Poisson--Lie type}. Now, one can construct a special class of background fields $(g,B,H)$ on the target $D/G$, using a fixed half-dimensional subspace $\E_{+} \subset \d$. See the text under equation (\ref{eq_V+asgB}) for details. 

Interestingly, if one repeats this procedure for a different Lagrangian subgroup $G' \subset D$ and constructs the respective fields $(g',B',H')$ starting from the same subspace $\E_{+}$, the corresponding $\sigma$-models on $D/G$ and $D/G'$ are equivalent. More precisely, there exists an (almost) symplectomorphism  of the respective phase spaces intertwining the Hamiltonians. See \cite{Severa:2017kcs} for details. This is the main statement of \textbf{Poisson--Lie T-duality}. 

On the other hand, for any $\sigma$-model (\ref{eq_SMaction}), one can consider the corresponding low-energy effective action 
\begin{equation} \label{eq_Seffaction}
\begin{split}
\S_{\text{eff}}[g,B,\phi] = \int_{M} e^{-2\phi} \{ & \RS(g) - \frac{1}{2}\<H + dB,H+dB\>_{g}\,+ \\
& + 4 \< d\phi,d\phi\>_{g} \} \cdot d\vol_{g},
\end{split}
\end{equation}
where $g$ and $B$ are now dynamical fields on $M$, and $\phi$ is a smooth function on $M$ called the dilaton field. Equivalently, this is a bosonic part of the type $\text{II}$ supergravity where the Ramond-Ramond fields are omitted. In this paper, however, the dimension of $M$ does not need to be ten. 

For some time, the theory of Courant algebroids, with a proper generalization of the Levi-Civita connection, seems to be the correct approach to a geometrical description of low-energy effective actions and various supergravities. See \cite{Coimbra:2011nw, Coimbra:2012af} for type \text{II} supergravities, \cite{Garcia-Fernandez:2013gja, Garcia-Fernandez:2016ofz} for the heterotic case, and our own work on this topic in \cite{Jurco:2015xra, Jurco:2015bfs, Jurco:2016emw, Vysoky:2017epf}. The same idea is pivotal in double field theory, for a comprehensive list of references see, e.g. the review in \cite{Hohm:2013bwa}. For a recent work on Poisson--Lie T-duality and related topics, see also \cite{Hassler:2016srl, Demulder:2018lmj, Hoare:2017ukq, Hoare:2018ebg}.

It is natural to combine the Courant algebroid approach to Poisson--Lie T-duality with the geometrical description of effective theories. 

We have done this in \cite{Jurco:2017gii} for a special case where $D$ is diffeomorphic to a product $G \times G^{\ast}$ of two mutually dual Poisson--Lie groups. The corresponding homogeneous spaces are $D/G \cong G^{\ast}$ and $D/G^{\ast} \cong G$ and this scenario in fact corresponds to the original setting of Poisson--Lie duality in \cite{Klimcik:1995ux}. Unfortunately, this paper contains a quite cumbersome derivation of the formula for the dilaton. Moreover, neither an explicit form of algebraic equations for the subspace $\E_{+} \subset \d$ nor examples are given. 

These issues are addressed in this paper. In Section \ref{sec_Maninpairs}, we recall the rich geometrical content of Manin pairs $(\d,\g)$ and their integration to Lie group pairs $(D,G)$. In particular, there is a natural structure of an exact Courant algebroid on the trivial vector bundle $D/G \times \d$ and a quasi-Poisson tensor on the coset space $D/G$. In Section \ref{sec_backgrounds}, we use these building blocks to construct background fields of the low-energy effective action corresponding to an arbitrary $\sigma$-model of Poisson--Lie type. In particular, one has to employ the apparatus of Courant algebroid connections to find the dilaton field $\phi$, c.f. Theorem \ref{thm_dilaton}. The proof of the resulting formula (\ref{eq_dilatonfinal}) is moved to Appendix \ref{ap_dilaton}. It is a great simplification and generalization to the quasi-Poisson case of the one presented in \cite{Jurco:2017gii}.

In Section \ref{sec_equations}, we prove the two main results of this paper. In Theorem \ref{thm_main2}, we show that equations of motion for the effective actions are, for background fields constructed in Section \ref{sec_backgrounds}, equivalent to a system of algebraic equations for the subspace $\E_{+} \subset \d$. As $\E_{+}$ is common to all $\sigma$-models related by (quasi-)Poisson--Lie $T$-duality, one immediately obtains Theorem \ref{thm_PLTdual}, which should be viewed as a proof of the consistency of Poisson--Lie T-duality with the induced low-energy theories. We prove the theorem using the language of Levi-Civita connections on Courant algebroids, see \cite{Jurco:2016emw}. The detailed derivation of the algebraic system of equations for $\E_{+}$ has been included as Appendix \ref{ap_eom}. 

At this moment, we must point out that recently a very interesting paper \cite{Severa:2018pag} appeared. Their main claim is very similar to our Theorem \ref{thm_PLTdual}. Moreover, Ramond-Ramond fields and more general $\sigma$-models are considered. They use methods developed in \cite{Severa:2016lwc}. In particular, instead of general Courant algebroid connections, they construct a (different) generalized Ricci scalar and a scalar curvature without using Levi-Civita connections. Instead they make use of a properly defined divergence operator $\Div: \Gamma(E) \rightarrow C^{\infty}(M)$. We believe that the approaches developed in \cite{Severa:2018pag} and in the present paper are complementary to each other; we write down explicit formulas in terms of the quasi-Poisson geometry on $D/G$, discuss the algebraic system of equations in more detail, and find some interesting non-trivial solutions for a type $\text{II}$ supergravity with a dilaton field in Section \ref{sec_examples}. 

\section{Manin pairs and their geometry} \label{sec_Maninpairs}
Let $\d$ be a Lie algebra together with a non-degenerate symmetric and invariant bilinear form $\<\cdot,\cdot\>_{\d}$. We will write $g_{\d} = \<\cdot,\cdot\>_{\d}$ whenever it is more convenient. For any subalgebra $\g \subseteq \d$, we say that $(\d,\g)$ is a \textbf{Manin pair} if $\g = \g^{\perp}$, that is, $\g$ is Lagrangian (maximally isotropic) with respect to $\<\cdot,\cdot\>_{\d}$. We assume that there are connected Lie groups $D$ and $G$ and such that $\d = \Lie(D)$ and $\g = \Lie(G)$, respectively, and such that $G \subset D$ forms a closed subgroup of $D$. For a detailed treatment of these topics, see \cite{Bursztyn:0310445,Bursztyn:0710.0639,Alekseev:0006168}

Let $\ti: \g \rightarrow \d$ denote the inclusion map, and let $\tq: \d \rightarrow \g^{\ast}$ be the quotient map $\d \rightarrow \d / \g$ composed with the canonical isomorphism $\d/\g \rightarrow \g^{\ast}$ induced by $\<\cdot,\cdot\>_{\d}$. By construction, we obtain a short exact sequence of vector spaces
\begin{equation} \label{eq_ManinSES}
\begin{tikzcd}
0 \arrow{r} & \g \arrow{r}{\ti} & \d \arrow{r}{\tq} & \g^{\ast} \arrow{r} \arrow[bend left=20, dashed]{l}{\tj} & 0.
\end{tikzcd}
\end{equation}
We say that $\tj: \g^{\ast} \rightarrow \d$ is an \textbf{isotropic splitting} of the sequence (\ref{eq_ManinSES}) if $\tq \circ \tj = 1$ and the subspace $\tj(\g^{\ast}) \subseteq \d$ is isotropic with respect to $\<\cdot,\cdot\>_{\d}$. There always exists such a splitting. Moreover, for any other isotropic splitting $\tj': \g^{\ast} \rightarrow \d$, there is a unique bivector $\theta \in \Lambda^{2} \g$ such that 
\begin{equation} \label{eq_splitchange}
\tj'(\xi) = \tj(\xi) + \ti( \theta(\xi)),
\end{equation}
where we always identify the bivector $\theta$ with the induced linear map $\xi \mapsto \theta(\cdot,\xi)$. The triple $(\d,\g,\tj)$ is called the \textbf{split Manin pair}. Each split Manin pair induces a unique \textbf{Lie quasi-bialgebra} $(\g,\delta,\mu)$ where $\delta: \g \rightarrow \Lambda^{2} \g$ and $\mu \in \Lambda^{3} \g$ are given by 
\begin{equation} \label{eq_defdeltamu}
\begin{split}
\delta(x)(\xi,\eta) = & \ \< [\tj(\xi), \tj(\eta)]_{\d}, \ti(x) \>_{\d}, \\ \mu(\xi,\eta,\zeta) = & \ \< [\tj(\xi),\tj(\eta)]_{\d}, \tj(\zeta) \>_{\d}. 
\end{split}
\end{equation}
Given a splitting $\tj$, one can construct a vector space isomorphism $\g \oplus \g^{\ast} \rightarrow \d$ and equip $\g \oplus \g^{\ast}$ with the structure of a \textbf{double of the Lie quasi-bialgebra $\g$}. It is uniquely determined by $(\g,\delta,\mu)$ and for all $(x,\xi), (y,\eta) \in \g \oplus \g^{\ast}$, it has the form
\begin{equation} \label{eq_doublebracket}
\begin{split}
[(x,\xi),(y,\eta)]_{\d} = \big( & [x,y]_{\g} + \ad^{\ast}_{\xi}(y) - \ad^{\ast}_{\eta}(x) + \mu(\xi,\eta,\cdot), \\
& [\xi,\eta]_{\g^{\ast}} + \ad^{\ast}_{x}(\eta) - \ad^{\ast}_{y}(\xi) \big),
\end{split}
\end{equation}
where $\< [\xi,\eta]_{\g^{\ast}}, x \> := \delta(x)(\xi,\eta)$ and $\ad^{\ast}_{\xi} := -[\xi,\cdot]_{\g^{\ast}}^{T}$. There are certain compatibility conditions among $(\g,\delta,\mu)$ which are most easily read out of the Jacobi identities for the bracket (\ref{eq_doublebracket}). Note that for $\mu \neq 0$, the Lagrangian subspace $\tj(\g^{\ast}) \subseteq \d$ is not a subalgebra and for general $\mu$, the skew-symmetric bracket $[\cdot,\cdot]_{\g^{\ast}}$ is not Lie. The Lie quasi-bialgebra $(\g,\delta',\mu')$corresponding to another isotropic splitting $\tj'$ can be expressed using $(\g,\delta,\mu)$ and the unique bivector $\theta$ by plugging (\ref{eq_splitchange}) into the definitions (\ref{eq_defdeltamu}).

As $G$ is assumed to be closed, there is a unique smooth manifold structure on the coset space $S = D / G$, making the quotient map $\pi_{0}: D \rightarrow S$ into a smooth surjective submersion and $D$ into a total space of a principal $G$-bundle with the Lie group $G$ acting on $D$ via the restriction of the right multiplication.

There is a natural transitive left action $\tr: D \times S \rightarrow S$ defined for all $d,k \in D$ by the formula
\begin{equation}
d \tr \pi_{0}(k) = \pi_{0}(dk), 
\end{equation}
called the \textbf{dressing action} of $D$ on $S$. Let $\#^{\tr}: \d \rightarrow \X(S)$ denote the corresponding infinitesimal action. It can be used to define a fiber-wise surjective vector bundle map $\rho: S \times \d \rightarrow TS$ given as $\rho(s,x) = \#^{\tr}_{s}(x)$. The trivial vector bundle $E = S \times \d$ can be equipped with a fiber-wise extension of the form $\<\cdot,\cdot\>_{\d}$, denoted by the same symbol. We may thus form a sequence of vector bundles
\begin{equation}
\begin{tikzcd} \label{eq_CourantSES}
0 \arrow{r} & T^{\ast}S \arrow{r}{\rho^{\ast}} & E \arrow{r}{\rho} & TS \arrow{r} \arrow[bend left=20, dashed]{l}{\sigma} & 0, 
\end{tikzcd}
\end{equation}
where $\rho^{\ast} = g_{\d}^{-1} \circ \rho^{T}$. Here $g_{\d}: E \rightarrow E^{\ast}$ denotes the vector bundle isomorphism induced by $g_{\d} = \<\cdot,\cdot\>_{\d}$. This sequence is exact, as this is in fact the Atiyah sequence for the principal bundle $\pi_{0}: D \rightarrow S$. See \cite{Jurco:2017gii} for details. $E$ can be equipped with a unique bracket $[\cdot,\cdot]_{E}$ extending the Lie bracket $-[\cdot,\cdot]_{\d}$ on constant sections such that $(E,\rho,\<\cdot,\cdot\>_{\d},[\cdot,\cdot]_{E})$ becomes an \textbf{exact Courant algebroid}. Set
\begin{equation} \label{eq_CourantbracketE}
\begin{split}
\<[\psi,\psi']_{E},\psi''\>_{\d} = & \ \< \Li{\rho(\psi)}(\psi') - \Li{\rho(\psi')}(\psi)\,- \\ 
& - [\psi,\psi']_{\d}, \psi''\>_{\d} + \< \Li{\rho(\psi'')}(\psi), \psi'\>_{\d},
\end{split}
\end{equation}
for all $\psi,\psi',\psi'' \in \Gamma(E) = C^{\infty}(S,\d)$. Note that the last term is the only difference between $[\cdot,\cdot]_{E}$ and the bracket corresponding to the respective Atiyah Lie algebroid.

For every exact Courant algebroid, there exists an isotropic splitting of the sequence (\ref{eq_CourantSES}), that is a vector bundle map $\sigma: TS \rightarrow E$ satisfying $\rho \circ \sigma = 1$ and $\sigma(TS) \subseteq E$ forming a Lagrangian subbundle with respect to $\<\cdot,\cdot\>_{\d}$. Every such $\sigma$ induces a Courant algebroid isomorphism $\fPsi_{\sigma}: \gTS \rightarrow E$ where $\gTS = TS \oplus T^{\ast}S$ is equipped with the \textbf{$H_{\sigma}$-twisted Dorfman bracket}
\begin{equation} \label{eq_HDorfman}
\begin{split}
[(X,\xi),(Y,\eta)]_{D}^{H_{\sigma}} = \big([X,Y], & \Li{X}(\eta) - \io_{Y}(d\xi)\,- \\
& - H_{\sigma}(X,Y,\cdot) \big),
\end{split}
\end{equation}
for all $X,Y \in \X(S)$ and $\xi,\eta \in \Omega^{1}(S)$. The anchor on $\gTS$ is the canonical projection onto $TS$ and the pairing $\<\cdot,\cdot\>_{\gTS}$ is the canonical one between $1$-forms and vector fields. The closed $3$-form $H_{\sigma}$ represents the so called \textbf{Ševera class of $E$} and it is obtained via the formula 
\begin{equation} \label{eq_severaclass}
H_{\sigma}(X,Y,Z) = -\< [\sigma(X),\sigma(Y)]_{E}, \sigma(Z) \>_{\d}, 
\end{equation}
for all $X,Y,Z \in \X(S)$. For any other isotropic splitting $\sigma'$ of (\ref{eq_CourantSES}), one has $[H_{\sigma}]_{dR} = [H_{\sigma'}]_{dR}$. For a good reference on the topic of exact Courant algebroids and their splittings, see e.g. \cite{Kotov:2010wr}. 

 Can some splittings $\tj: \g^{\ast} \rightarrow \d$ of (\ref{eq_ManinSES}) be of use in order to construct splittings $\sigma$ of the short exact sequence (\ref{eq_CourantSES})? As $\sigma$ is uniquely determined by its image, it is natural to consider a subbundle $S \times \tj(\g^{\ast}) \subseteq E$. This is a Lagrangian subbundle of a correct rank. One only has to show that it is complementary to $\ker(\rho)$. Unfortunately, this is not true for general $\tj$. First, note that for any $s \in S$, one can unambiguously define a subspace $\Ad_{s}(\g) := \Ad_{d}(\g)$ for any $d \in \pi_{0}^{-1}(s)$. One says that the isotropic splitting $\tj$ of (\ref{eq_ManinSES}) is \textbf{admissible at $s \in S$}, if 
\begin{equation} \label{eq_admissible}
\d = \Ad_{s}(\g) \oplus \tj(\g^{\ast}).
\end{equation}
Every splitting is admissible at $s_{0} = \pi_{0}(G)$. If $\tj$ is admissible at $s$, it is admissible at all points of some neighborhood of $s$. Finally, for every $s \in S$, there exists some splitting admissible at $s$. For the proof of the last assertion, see \cite{Alekseev:0006168}. We say that $(D,G)$ is a \textbf{complete group pair} if it admits an \textbf{everywhere admissible splitting}. Then the following statements are equivalent:
\begin{enumerate}[i)]
\item $\tj: \g^{\ast} \rightarrow \d$ is everywhere admissible;
\item The subbundle $S \times \tj(\g^{\ast})$ is complementary to $\ker(\rho)$. In other words, there exists a unique isotropic splitting $\sigma: TS \rightarrow E$ of (\ref{eq_CourantSES}), such that $\sigma(TS) = S \times \tj(\g^{\ast})$;
\item For each $s \in S$, the map $\xi \mapsto \#^{\tr}_{s}(\tj(\xi))$ is a linear isomorphism. In other words, the module $\X(S)$ is generated by vector fields $\xi^{\tr} := \#^{\tr}(\tj(\xi))$. The splitting $\sigma$ from the previous point can be then uniquely described by 
\begin{equation} \label{eq_sigmaasxitr}
\sigma(\xi^{\tr}) = \tj(\xi), 
\end{equation}
where we identify elements of $\d$ with the constant sections of $E$;
\item For a splitting $\tj$, one can write the adjoint action $\Ad$ of $D$ as a formal block matrix with respect to the isomorphism $\d \cong \g \oplus \g^{\ast}$ induced by the choice of $\tj$:
\begin{equation} \label{eq_Adblock}
\Ad_{d} = \bma{\fk(d)}{\fb(d)}{\fc(d)}{\fa(d)},
\end{equation}
where $\fk(d): \g \rightarrow \g$ and similarly for the other blocks. For everywhere admissible $j$, $\fk(d)$ \textit{is invertible} for all $d \in D$. 
\end{enumerate}
Note, it follows immediately from (\ref{eq_severaclass}) and (\ref{eq_sigmaasxitr}) that $H_{\sigma}$ is related to $\mu \in \Lambda^{3} \g$ by 
\begin{equation} \label{eq_Hinadmissiblesplit}
H_{\sigma}( \xi^{\tr}, \eta^{\tr}, \zeta^{\tr}) = \mu(\xi,\eta,\zeta),
\end{equation}
for all $\xi,\eta,\zeta \in \g^{\ast}$. This observation underlines the general principle - evaluate everything on the special vector fields $\xi^{\tr}$ to make the calculations easier.  

Finally, the choice of an isotropic splitting $\tj: \g^{\ast} \rightarrow \d$ allows one to induce additional structure on $D$. Indeed, consider the tensor $r_{\tj} \in \d \otimes \d$ defined by 
\begin{equation} \label{eq_standardr}
r_{\tj}( \xi,\eta) = \< \ti^{T}(\xi), \tj^{T}(\eta) \>,
\end{equation}
for all $\xi,\eta \in \d^{\ast}$. This is called the \textbf{standard $r$-matrix corresponding to $\tj$}. Let 
\begin{equation} \label{eq_PiDtensor}
\Pi^{D}_{\tj} := r_{\tj}^{L} - r_{\tj}^{R},
\end{equation}
where $L$ and $R$ denote the left and right translation of an element of $\d \otimes \d$ along the group $D$, respectively. One can show that $\Pi^{D}_{\tj} \in \X^{2}(D)$ forms a multiplicative bivector field on $D$, that is
\begin{equation}
( \Pi_{\tj}^{D})_{hk} = L_{h \ast}(\Pi_{\tj}^{D})_{k} + R_{k \ast}( \Pi^{D}_{\tj})_{h} ,
\end{equation}
for all $h,k \in D$. It is not a Poisson tensor though. Let $\mu_{\d} := (\Lambda^{3} \ti)(\mu)$ be the $3$-vector $\mu$ viewed as an element of $\Lambda^{3} \d$. If $[\cdot,\cdot]$ denotes the usual Schouten-Nijenhuis bracket, we obtain 
\begin{equation}
\frac{1}{2} [\Pi^{D}_{\tj}, \Pi^{D}_{\tj}] = \mu^{L}_{\d} - \mu^{R}_{\d}, \; \; [\Pi^{D}_{\tj}, \mu^{L}_{\d} ] = [\Pi^{D}_{\tj}, \mu^{R}_{\d}] = 0. 
\end{equation}
The tensor $\Pi^{D}_{\tj}$ depends on the choice of the splitting $\tj$. When $\tj'$ is related to $\tj$ as in (\ref{eq_splitchange}), one finds 
$\Pi_{\tj'}^{D} = \Pi_{\tj}^{D} + \theta_{\d}^{L} - \theta_{\d}^{R}$, where $\theta_{\d} = (\Lambda^{2} \ti)(\theta)$. There is a well-defined bivector $\Pi^{S}_{\tj} \in \X^{2}(S)$, such that $\pi_{0 \ast}( \Pi^{D}_{\tj}) = \Pi^{S}_{\tj}$. It can be written directly as $\Pi^{S}_{\tj} = - (\Lambda^{2} \#^{\tr})(r_{\tj})$. From the above equations for $\Pi_{D}$, one can directly derive the identities 
\begin{equation}
\frac{1}{2}[ \Pi^{S}_{\tj}, \Pi^{S}_{\tj}] = -(\Lambda^{3} \#^{\tr})(\mu_{\d}), \; \; [\Pi^{S}_{\tj}, (\Lambda^{3} \#^{\tr})(\mu_{\d})] = 0. 
\end{equation}
Moreover, under the change of splitting (\ref{eq_splitchange}), one has 
\begin{equation} \Pi^{S}_{\tj'} = \Pi^{S}_{\tj} - (\Lambda^{2} \#^{\tr})(\theta_{\d}). \end{equation}
In the following, we will assume a fixed splitting $\tj$ and omit the corresponding subscript. 

Now, suppose that $\tj$ is everywhere admissible. For each $x \in \g$, we may define a $1$-form $x^{\tr} \in \Omega^{1}(S)$ by requiring $x^{\tr}( \xi^{\tr}) = \xi(x)$ for all $\xi \in \g^{\ast}$. This determines it uniquely. The infinitesimal action $\#^{\tr}$ can be then written as 
\begin{equation}
\#^{\tr}( \ti(x) + \tj(\xi)) = \xi^{\tr} - \Pi^{S}(x^{\tr}). 
\end{equation}
Equivalently, $\Pi^{S}$ can be uniquely characterized by equation $\Pi^{S}(x^{\tr},y^{\tr}) = y^{\tr}(\#^{\tr}(\ti(x)))$. To save some space, define a function $\fPi \in C^{\infty}(S, \Lambda^{2} \g^{\ast})$ for all $x,y \in \g$ as 
\begin{equation} \label{eq_fPi}
\fPi(x,y) = \Pi^{S}(x^{\tr},y^{\tr}). 
\end{equation}
We will often view $\fPi$ as an $S$-dependent map from $\g$ to $\g^{\ast}$ defined by $\fPi(y) = \fPi(\cdot,y)$. The objects introduced in this paragraph satisfy some important relations which are in detail discussed in Appendix \ref{ap_relations}. There is a non-trivial observation relating $\fPi$ to the block form of the adjoint representation (\ref{eq_Adblock}). We formulate it as a proposition.
\begin{tvrz} \label{tvrz_fPiasck}
Let $\fk: D \rightarrow \Aut(\g)$ and $\fc: D \rightarrow$\linebreak $ \Hom(\g, \g^{\ast})$ be the smooth maps defined by the block decomposition (\ref{eq_Adblock}). 
Then the map $\fPi \in C^{\infty}(S, \Hom(\g,\g^{\ast}))$ introduced above can be written as $\fPi \circ \pi_{0} = \fc \cdot \fk^{-1}$.
\end{tvrz}
\begin{proof}
The map $\fk$ can be point-wise inverted as $\tj$ is assumed to be everywhere admissible. As $\Ad$ is a group representation and $G$ is a subgroup, the function $\fc \cdot \fk^{-1} \in C^{\infty}(D, \Hom(\g,\g^{\ast}))$ is $G$-invariant. Hence $\fPi' \circ \pi_{0} = \fc \cdot \fk^{-1}$ for some $\fPi' \in C^{\infty}(S, \Hom(\g,\g^{\ast}))$. One only has to argue that $\fPi' = \fPi$. To achieve this, one has to observe that $\pi_{0}^{\ast}(x^{\tr}) = \{ \tq^{T}( \fk^{-1}x) \}_{L}$. The rest follows from the fact that $\Pi^{S} = \pi_{0 \ast}(\Pi^{D})$ and the definitions (\ref{eq_standardr}, \ref{eq_PiDtensor}). 
\end{proof}
\section{Constructing the background fields} \label{sec_backgrounds}
We will now construct $\sigma$-model backgrounds $(g,B,H)$ on the manifold $S$ starting from two pieces of the algebraical data on the Lie algebra $\d$. 

First, recall that by \textbf{generalized metric} on any orthogonal vector bundle $(E, \<\cdot,\cdot\>_{E})$ we mean a maximal positive subbundle $V_{+} \subseteq E$ with respect to $\<\cdot,\cdot\>_{E}$. Its rank equals to the positive index (which has to be constant on every connected component of the base manifold) of the form $\<\cdot,\cdot\>_{E}$. We define $V_{-} = V_{+}^{\perp}$. It follows that $V_{-}$ is a maximal negative subbundle of $E$ with respect to $\<\cdot,\cdot\>_{E}$ and we can write $E = V_{+} \oplus V_{-}$. 

In particular, on the vector bundle $\gTS = TS \oplus T^{\ast}S$ equipped with the canonical pairing $\<\cdot,\cdot\>_{\mathbb{T}}$, every generalized metric $V_{+} \subseteq \mathbb{T}S$ is uniquely determined by a pair $(g,B)$, where $g$ is a Riemannian metric on $S$ and $B \in \Omega^{2}(S)$ is a $2$-form. More precisely, one has 
\begin{equation} \label{eq_V+asgB}
\Gamma(V_{+}) = \{ (X,(g+B)(X)) \; | \; X \in \X(S) \}.
\end{equation}

We use this observation to construct the backgrounds. This is the original idea of Poisson--Lie T-duality \cite{Klimcik:1995ux}, explained in the language of Courant algebroids in \cite{Severa:2015hta}. Fix a maximal positive subspace $\E_{+} \subset \d$ with respect to $\<\cdot,\cdot\>_{\d}$. Note that the positive index of $\<\cdot,\cdot\>_{\d}$ is always $\dim(\g)$ as $\g \subset \d$ is assumed Lagrangian and thus $\g = \g^{\perp}$. Take the trivial positive subbundle $\E_{+}^{E} = S \times \E_{+} \subset E$. Obviously, $\E_{+}^{E}$ is a generalized metric on the orthogonal vector bundle $(E, \<\cdot,\cdot\>_{\d})$. 

One can fix any splitting $\sigma: TS \rightarrow E$ of the sequence (\ref{eq_CourantSES}) and obtain a vector bundle isomorphism $\fPsi_{\sigma}: \gTS \rightarrow E$. Let $V_{+}^{\sigma} = \fPsi_{\sigma}^{-1}( \E^{E}_{+})$. This is, by construction, a generalized metric on $(\gTS, \<\cdot,\cdot\>_{\mathbb{T}})$ and there is thus a unique pair $(g, B_{\sigma})$ of background fields determined by $V_{+}^{\sigma}$. The subscript $\sigma$ of $g$ is missing on purpose, the Riemannian metric obtained in this ways is in fact independent on the splitting. 

Recall that there is also a closed $3$-form $H_{\sigma} \in \Omega^{3}(S)$ representing the Ševera class of $E$ defined by (\ref{eq_severaclass}). Altogether, we obtain a triple $(g,B_{\sigma},H_{\sigma})$ of sigma model background fields on the target space $S$. These are precisely the ones discussed in \cite{Severa:2017kcs}. The choice of the splitting $\sigma$ is not especially important. Indeed, every other isotropic splitting $\sigma': TS \rightarrow E$ of (\ref{eq_CourantSES}) can be written as 
\begin{equation} \label{eq_changeofsplitting}
\sigma'(X) = \sigma(X) + \rho^{\ast}( C(X)) 
\end{equation}
for a unique $2$-form $C \in \Omega^{2}(S)$. For the corresponding $3$-form $H_{\sigma'}$, one has $H_{\sigma'} = H_{\sigma} + dC$, whereas the $2$-form $B_{\sigma'}$ is given by $B_{\sigma'} = B_{\sigma} - C$. Using the Stokes theorem, it is clear that the two sigma models with target space backgrounds $(g,B_{\sigma},H_{\sigma})$ and $(g, B_{\sigma} - C, H_{\sigma} + dC)$ are equivalent. We can thus choose $\sigma$ to our advantage in the following. Let us also henceforth write just $B$ for $B_{\sigma}$. 

Let $\tj$ be an everywhere admissible splitting. There is a unique invertible map $E_{0}: \g^{\ast} \rightarrow \g$ such that the positive definite subbundle $\E_{+} \subseteq \d$ can be written as its graph:
\begin{equation} \label{eq_defE0map}
\E_{+} = \{ \ti( E_{0}(\xi)) + \tj( \xi) \; | \; \xi \in \g^{\ast} \}.
\end{equation}
The subspace $\E_{+}$ is positive with respect to $\<\cdot,\cdot\>_{\d}$ if and only if $E_{0}$ has a positive definite symmetric part. Define $\fE \in C^{\infty}(S, \Hom(\g^{\ast},\g))$ for all $\xi,\eta \in \g^{\ast}$ using the equation
\begin{equation}
\< \xi, \fE(\eta) \> = \< \xi^{\tr}, (g + B)( \eta^{\tr}) \>.
\end{equation}
In other words, $\fE$ is the combination of the metric and $B$-field in the special frame. Again, it makes sense to view $\fE$ as an $S$-dependent map from $\g^{\ast}$ to $\g$. It follows that $\fE$ can be written in terms of the function $\fPi$ given by (\ref{eq_fPi}) and the constant map $E_{0}$ defined above as
\begin{equation} \label{eq_fEasfE0Pi}
\fE = ( E_{0}^{-1} - \fPi)^{-1}.
\end{equation}
To see this, note that $\rho^{\ast}(x^{\tr}) = \ti(x) + \tj( \fPi(x))$ for all $x \in \g$. The section $(\eta^{\tr}, \fE(\eta)^{\tr}) \in \Gamma(V^{\sigma}_{+})$ is by $\fPsi_{\sigma}: \gTS \rightarrow E$ mapped onto the section of $\E_{+}^{E}$. But
\begin{equation}
\begin{split}
\fPsi_{\sigma}( \eta^{\tr}, \fE(\eta)^{\tr}) = & \  \sigma(\eta^{\tr})  + \rho^{\ast}( \fE(\eta)^{\tr}) \\
= & \ \ti( \fE(\eta)) + \tj( \eta + \fPi( \fE(\eta))).
\end{split}
\end{equation}
By definition of $E_{0}$, we find $\fE(\eta) = E_{0}(\eta + \fPi(\fE(\eta)))$ for all $\eta \in \g^{\ast}$. It is not difficult to see that $\fE$ given by (\ref{eq_fEasfE0Pi}) is the unique solution of this equation.

Now, let us endeavor to find the last of the background fields in the effective action, the \textit{dilaton field} $\phi \in C^{\infty}(S)$. To do so, we have to introduce another concept. 

Let $(E,\rho, \<\cdot,\cdot\>_{E}, [\cdot,\cdot]_{E})$ be a Courant algebroid over a manifold $M$. We say that an $\R$-bilinear map  $\cD: \Gamma(E) \times \Gamma(E) \rightarrow \Gamma(E)$ is a \textbf{Courant algebroid connection}, if the $\R$-linear operator $\cD_{\psi} := \cD(\psi,\cdot)$ on $\Gamma(E)$ satisfies
\begin{equation}
\begin{split}
\cD_{\psi}(f \psi') = & \ f \cD_{\psi}(\psi') + \Li{\rho(\psi)}(f) \psi', \\
\cD_{f \psi}(\psi') = & \ f \cD_{\psi}(\psi'),
\end{split}
\end{equation}
for all sections $\psi,\psi' \in \Gamma(E)$ and $f \in C^{\infty}(M)$, and its extension to  tensors on $E$ satisfies $\cD_{\psi}(g_{E}) = 0$. We write $g_{E} = \<\cdot,\cdot\>_{E}$. One says that $\cD$ is \textbf{torsion-free} if its torsion $3$-form defined for all $\psi,\psi',\psi'' \in \Gamma(E)$ by 
\begin{equation}
\begin{split}
T_{\cD}(\psi,\psi',\psi'') = & \ \< \cD_{\psi}(\psi') - \cD_{\psi'}(\psi) - [\psi,\psi']_{E}, \psi''\>_{E} \\
& + \<\cD_{\psi''}(\psi),\psi'\>_{E} 
\end{split}
\end{equation}
vanishes identically. Finally, let $V_{+} \subseteq E$ be a generalized metric. We say that $\cD$ is a \textbf{Levi-Civita connection on $E$ with respect to $V_{+}$}, if it is torsion-free and for all $\psi \in \Gamma(E)$, one has $\cD_{\psi}(\Gamma(V_{+})) \subseteq \Gamma(V_{+})$. One writes $\cD \in \LC(E,V_{+})$. The space of such connections is non-empty and in general quite big, see e.g. \cite{Jurco:2016emw} or \cite{Garcia-Fernandez:2016ofz}. For any Courant algebroid connection, there is a natural divergence operator $\Div_{\cD}: \Gamma(E) \rightarrow C^{\infty}(M)$ defined as $\Div_{\cD}(\psi) = \Tr( \cD(\cdot,\psi))$. By $\Tr$ we mean the fiber-wise trace of the vector bundle endomorphism. 

Now, recall that our starting data in the construction of the background $(g,B)$ was a generalized metric $\E_{+} \subseteq \d$. On top of that, consider a Levi-Civita connection $\cD^{0} \in \LC(\d, \E_{+})$, where $\d$ is viewed as a Courant algebroid $(\d,0,\<\cdot,\cdot\>_{\d}, -[\cdot,\cdot]_{\d})$ over the point. Let $E = S \times \d$ be the Courant algebroid with the bracket (\ref{eq_CourantbracketE}). It restricts to the bracket $-[\cdot,\cdot]_{\d}$ on constant sections and we can thus extend $\cD^{0}$ to a unique Courant algebroid connection $\cD^{E} \in \LC(E, \E^{E}_{+})$. Finally, by fixing the splitting $\sigma: TS \rightarrow E$, one can use the isomorphism $\fPsi_{\sigma}: \gTS \rightarrow E$ to define a Levi-Civita connection $\cD^{\sigma} \in \LC(\gTS, V_{+}^{\sigma})$. 

For reasons explained in the following section, we \textit{define} $\phi \in C^{\infty}(S)$ to be the solution to the system of partial differential equations
\begin{equation} \label{eq_dilatonformula}
\Li{X}( \phi) = \frac{1}{2}(\Div_{\cD^{g}}(X) - \Div_{\cD^{\sigma}}(X,\xi)),
\end{equation}
for all $(X,\xi) \in \Gamma(\gTS)$. $\Div_{\cD^{g}}(X) = \Tr(\cD^{g}(\cdot,X))$ denotes the usual divergence of $X \in \X(S)$ induced by the ordinary Levi-Civita connection $\cD^{g}$ on $S$ corresponding to the metric $g$. Equivalently, this ensures the partial integration rule in the form
\begin{equation}
\begin{split}
\int_{S}  e^{-2\phi} \{ \Div_{\cD^{\sigma}}(X,\xi) \cdot f \} \; d\vol_{g}= \\
= - \int_{S} e^{-2\phi} \Li{X}(f) \; d\vol_{g}.
\end{split}
\end{equation}
If such $\phi$ exists, it is determined uniquely up to an additive constant, which is irrelevant for the effective action. It turns out that there are choices of $\cD^{0}$ for which there are no solutions. In fact, there even exist Manin pairs where there is no solution for any $\cD^{0}$. Fortunately, there are two reasonable restrictions on $(\d,\g)$ and $\cD^{0}$ where not only $\phi$ can be found, but it also can be written down explicitly. We state the result as a theorem. We have included its proof in Appendix \ref{ap_dilaton}.

\begin{theorem} \label{thm_dilaton}
Let $\g$ be a unimodular Lie algebra, that is $\Tr(\ad_{x}) = 0$ for all $x \in \g$. Suppose $\cD^{0} \in \LC(\d,\E_{+})$ is divergence-free. Then there exist a solution $\phi \in C^{\infty}(S)$ of the equation (\ref{eq_dilatonformula}), unique up to an additive constant. 

The equation (\ref{eq_dilatonformula}) is in fact independent of the used divergence-free connection $\cD^{0}$ as well as on the splitting $\sigma: TS \rightarrow E$. Moreover, if $\tj: \g^{\ast} \rightarrow \d$ is a splitting admissible on an open set $U \subseteq S$, the solution is on $U$ is given by the formula 
\begin{equation} \label{eq_dilatonfinal}
\phi = - \frac{1}{2} \ln( \det(\f1_{\g} - E_{0} \fPi)) - \frac{1}{2} \ln( \bm{\nu}).
\end{equation}
The function $\bm{\nu}: U \rightarrow \R^{+}$ is defined via the expression $\bm{\nu} \circ \pi_{0} = \det(\fk)$, where $\fk: \pi_{0}^{-1}(U) \rightarrow \Aut(\d)$ is the top-left block in the decomposition (\ref{eq_Adblock}). In fact, the right-hand side of (\ref{eq_dilatonfinal}) does not depend on the used admissible splitting. In particular, it can be used to define a globally well-defined solution $\phi \in C^{\infty}(S)$. 
\end{theorem}
\begin{rem}
There are several ways how to rewrite the formula (\ref{eq_dilatonfinal}). Recall that we have assumed that $G$ is unimodular. Moreover, $\d$ is a quadratic Lie algebra, which implies that also $D$ is unimodular. This ensures that there exists a unique (up to a constant) $D$-invariant volume form $\mu$ on $S$. If $\tj$ is a splitting admissible on some open set $U \subseteq S$, one can write this volume form as
\begin{equation} \label{eq_muform}
\mu = \bm{\nu} \cdot t_{1}^{\tr} \^ \ldots \^ t_{\dim(\g)}^{\tr},
\end{equation}
where $(t_{k})_{k=1}^{\dim(\g)}$ is some basis of $\g$. The form (\ref{eq_muform}) is indeed invariant with respect to the dressing action. To see this, note that $\mu$ is $D$-invariant if and only if its pullback $\pi_{0}^{\ast}(\mu) \in \Omega^{1}(D)$ is left-invariant. But we have already noted in the proof of Proposition (\ref{tvrz_fPiasck}) that 
\begin{equation}
\pi_{0}^{\ast}(x^{\tr}) = \{ \tq^{T}( \fk^{-1}x) \}_{L},
\end{equation}
for all $x \in \g$. Recall that $\fk: D \rightarrow \End(\g)$ is defined by (\ref{eq_Adblock}). Moreover, recall that $\bm{\nu} \circ \pi_{0} = \det(\fk)$. Altogether, one has 
\begin{equation}
\begin{split}
\pi_{0}^{\ast}(\mu) = & \ \det(\fk) \cdot \{ \tq^{T}(\fk^{-1}t_{1}) \}_{L} \^ \ldots \^ \{ \tq^{T}(\fk^{-1}t_{\dim(\g)}) \}_{L} \\
= & \ \{ \tq^{T}(t_{1}) \}_{L} \^ \ldots \^ \{ \tq^{T}(t_{\dim(\g)}) \}_{L}. 
\end{split}
\end{equation}
This form is indeed left-invariant. Moreover, note that $\mu$ does not depend on the choice of the admissible isotropic splitting $\tj$. In other words, the local definitions (\ref{eq_muform}) patch well to give a global form $\mu$ on $S$. Now, let $\fE = \fg + \fB$ be the decomposition of the map (\ref{eq_fEasfE0Pi}) into its symmetric and skew-symmetric part. Slightly abusing the notation, we use the same notation for the corresponding matrices in the basis we have used to define $\mu$. Then
\begin{equation}
\det(\fE) = \det(\fg)^{\frac{1}{2}} \det( \fg - \fB \fg^{-1} \fB)^{\frac{1}{2}}.
\end{equation}
Now, note that (\ref{eq_fEasfE0Pi}) implies $\fg - \fB \fg^{-1} \fB = g_{0}^{-1} - \theta_{0} g_{0} \theta_{0}$, where $E_{0} = g_{0}^{-1} + \theta_{0}$. Using the same formula for the determinant, we also have
\begin{equation}
\det(E_{0}) = \det(g_{0})^{-\frac{1}{2}} \det(g_{0}^{-1} - \theta_{0} g_{0} \theta_{0})^{\frac{1}{2}}. 
\end{equation}
Hence, the combination of these two expressions yields
\begin{equation}
\det(\f1_{\g} - E_{0} \fPi) = \det(E_{0} \fE^{-1}) = \det(g_{0})^{-\frac{1}{2}} \det(\fg)^{-\frac{1}{2}}. 
\end{equation}
Now, observe that the volume form $\vol_{g}$ can be written using the matrix $\fg$ of functions as 
\begin{equation}
\vol_{g} = \det(\fg)^{\frac{1}{2}} \cdot t_{1}^{\tr} \^ \ldots \^ t_{\dim(\g)}^{\tr}. 
\end{equation}
We can thus rewrite the formula (\ref{eq_dilatonfinal}) in the form 
\begin{equation} \label{eq_phialternative}
\phi = \frac{1}{2} \ln( \frac{\vol_{g}}{\mu}) + \frac{1}{4} \ln(\det(g_{0})). 
\end{equation}
Note that $\mu$ depends on the choice of the basis $(t_{k})_{k=1}^{\dim(\g)}$ which is compensated by the constant term in (\ref{eq_phialternative}). Moreover, none of the objects depend on the choice of the admissible splitting $\tj$.
\end{rem}

\section{Equations of motion for effective actions} \label{sec_equations}

Now, suppose that we have constructed the background fields $(g,B,\phi,H)$ as described in the previous section. We will now prove that the equations of motion of the action
\begin{equation} \label{eq_effaction}
\begin{split}
\S_{\text{eff}}[g,B,\phi] = \int_{S} e^{-2 \phi} \{ & \RS(g) - \frac{1}{2} \<H+dB,H+dB\>_{g}\,+ \\
& + 4\<d\phi,d\phi\>_{g} \} \cdot d\vol_{g}
\end{split}
\end{equation}
can be in this case formulated as a set of algebraic equations for the subspace $\E_{+} \subseteq \d$ we have used to construct this special class of target space fields. We will rely on the geometrical description of equations of motion in terms of Courant algebroid connections. 

Let us recall the necessary ingredients. Suppose we have a Courant algebroid $(E, \rho, \<\cdot,\cdot\>_{E}, [\cdot,\cdot]_{E})$ over a manifold $M$. Let $\cD$ be any Courant algebroid connection on $E$. For all $\psi,\psi',\phi,\phi' \in \Gamma(E)$, define
\begin{equation}
\begin{split}
R^{(0)}_{\cD}(\phi',\phi,\psi,\psi') = \< \phi', & \cD_{\psi}(\cD_{\psi'}(\phi)) - \cD_{\psi'}(\cD_{\psi}(\phi))\,- \\
& - \cD_{[\psi,\psi']_{E}}(\phi) \>_{E}.
\end{split}
\end{equation}
Then the \textbf{generalized Riemann tensor} $R_{\cD}$ is defined by
\begin{equation}
\begin{split}
2 R_{\cD}(\phi',\phi,\psi,\psi') = & \ R_{\cD}^{(0)}(\phi',\phi,\psi,\psi')\,+  \\
& + R_{\cD}^{(0)}(\psi',\psi,\phi,\phi')\,+ \\
& + \< \fK(\psi,\psi'), \fK(\phi,\phi') \>_{E},
\end{split}
\end{equation}
where $\fK(\psi,\psi') \in \Gamma(E)$ is defined by $\< \fK(\psi,\psi'), \phi \>_{E} = \< \cD_{\phi}(\psi),\psi'\>_{E}$. It turns out that $R_{\cD}$ is a well-defined tensor on $E$ with many useful symmetries, see \cite{Jurco:2016emw} for the detailed exposition. Most importantly, it allows one to unambiguously define the \textbf{generalized Ricci tensor} $\Ric_{\cD}$ 
\begin{equation}
\Ric_{\cD}(\psi,\psi') = R_{\cD}( \psi^{\mu}_{E}, \psi, \psi_{\mu}, \psi'),
\end{equation}
for all $\psi,\psi' \in \Gamma(E)$, where $(\psi_{\mu})_{\mu=1}^{\rk(E)}$ is some local frame on $E$ and $\< \psi^{\mu}_{E}, \psi_{\nu} \>_{E} = \delta^{\mu}_{\nu}$. Now, the choice of a generalized metric $V_{+} \subseteq E$ allows one to write $E = V_{+} \oplus V_{-}$. It follows that the fiber-wise metric $\gm = g_{E}|_{V_{+}} - g_{E}|_{V_{-}}$ is positive-definite. We can use it to define the \textbf{generalized scalar curvature} $\RS_{\cD}^{+} \in C^{\infty}(M)$ with respect to $V_{+}$ as a trace 
\begin{equation} \label{eq_pluscurvature}
\RS_{\cD}^{+} = \Tr_{\gm}( \Ric_{\cD}). 
\end{equation}
There is also a canonical scalar curvature $\RS_{\cD} = \Tr_{g_{E}}(\Ric_{\cD})$ which can be used for various consistency checks. These are precisely the ingredients required to describe the equations of motion of (\ref{eq_effaction}). We formulate it as a theorem. Its proof can be found in \cite{Jurco:2016emw}. 
\begin{theorem} \label{thm_main}
Let $g$ be any Riemannian metric on an arbitrary manifold $M$, let $B \in \Omega^{2}(M)$ be a $2$-form and $H \in \Omega^{3}(M)$ a closed $3$-form. Let $E = \mathbb{T}M$ be equipped with the Courant algebroid structure of $H$-twisted Dorfman bracket (\ref{eq_HDorfman}). Let $V_{+} \subseteq \gTM$ be a generalized metric encoding $(g,B)$ like in (\ref{eq_V+asgB}). 

Let $\cD \in \LC(\gTM,V_{+})$ be a Levi-Civita connection on $\gTM$ with respect to $V_{+}$. Suppose $\cD$ satisfies the additional condition
\begin{equation} \label{eq_dilatonformulajj}
\Li{X}( \phi) = \frac{1}{2}(\Div_{\cD^{g}}(X) - \Div_{\cD}(X,\xi)),
\end{equation}
for a given smooth function $\phi \in C^{\infty}(M)$ and all $(X,\xi) \in \Gamma(\gTM)$. We write $\cD \in \LC(\gTM,V_{+},\phi)$. 

Then the backgrounds $(g,B,\phi)$ solve the equations of motion obtained by the  variation of the action\footnote{The integration is now assumed over $M$.} (\ref{eq_effaction}) if and only if the connection $\cD$ satisfies the conditions $\RS_{\cD}^{+} = 0$ and $\Ric_{\cD}(V_{+},V_{-}) = 0$. 
\end{theorem}
It is of course important that for any $V_{+}$ and any $\phi \in C^{\infty}(M)$, we have $\LC(\gTM, V_{+},\phi) \neq \emptyset$. It follows directly from the statement that nothing depends on the particular choice of $\cD \in \LC(\gTM,V_{+},\phi)$. In fact, one can prove that  for any Levi-Civita connection, the quantities $\RS_{\cD}^{+}$, $\RS_{\cD}$ and $\Ric_{\cD}^{+-}$ depend on $\cD$ only through its divergence operator. As we fix its form by condition (\ref{eq_dilatonformulajj}), this explains the above observation.

Now, recall that we have constructed the dilaton function $\phi$ as a solution to the partial differential equation (\ref{eq_dilatonformula}). Looking at the condition (\ref{eq_dilatonformulajj}), we just made our Levi-Civita connection $\cD^{\sigma}$ satisfy the assumptions of Theorem \ref{thm_main}. Recall, we have constructed it from a given Levi-Civita connection $\cD^{0} \in \LC(\d,\E_{+})$. Hence, it is not surprising that their respective induced quantities are closely related. 

\begin{theorem} \label{thm_Requivalences}
Let $\cD^{0} \in \LC(\d,\E_{+})$ be any Levi-Civita connection on the Courant algebroid $(\d,0,\<\cdot,\cdot\>_{\d},-[\cdot,\cdot]_{\d})$ with respect to the generalized metric $\E_{+} \subseteq \d$. Let $\cD^{\sigma} \in \LC(\gTS, V_{+}^{\sigma})$ be the Levi-Civita connection constructed from it using the splitting $\sigma$, cf. the paragraph above (\ref{eq_dilatonformula}). 

Then $\RS^{+}_{\cD^{\sigma}} = \RS_{\cD^{0}}^{+}$ and $\Ric_{\cD^{\sigma}}(V_{+}^{\sigma},V_{-}^{\sigma}) = 0$, if and only if $\Ric_{\cD^{0}}(\E_{+},\E_{-}) = 0$. 
\end{theorem}
\begin{proof}
This is a slight generalization of Proposition 6.1 in \cite{Jurco:2017gii}. 
Note that in particular, the scalar curvature $\RS^{+}_{\cD^{\sigma}} \in C^{\infty}(S)$ is just a real constant. For the proof recall, that both $V_{+}^{\sigma}$ and $\cD^{\sigma}$ were constructed from $\E_{+}^{E}$ and $\cD^{E}$ using the Courant algebroid isomorphism $\fPsi_{\sigma}: \gTS \rightarrow E$. 

It is easy to see that both $\Ric$ and $\RS^{+}$ transform naturally with respect to any Courant algebroid isomorphism, respectively. Again, see \cite{Jurco:2016emw} for details. 

Finally, we have defined $\cD^{E}$ and $\E_{+}^{E}$ to restrict to $\cD^{0}$ and $\E_{+}$ on constant sections which (globally) generate $E$. As $R_{\cD^{E}}$ is a tensor on $E$, it (and all its traces) can be evaluated on the generators and the conclusion of the theorem follows immediately. 
\end{proof}
We can now formulate and easily prove the two main results of this paper. 

\begin{theorem} \label{thm_main2}
Let $\cD^{0} \in \LC(\d,\E_{+})$ be a divergence-free Levi-Civita connection. Suppose $\g \subseteq \d$ is a unimodular Lie algebra. Let $(g,B,\phi)$ be the background on $S$ constructed in Section \ref{sec_backgrounds}, cf. (\ref{eq_fEasfE0Pi}, \ref{eq_dilatonfinal}).

Then $(g,B,\phi)$ solve the equations of motion given by the action (\ref{eq_effaction}), if and only if $\RS^{+}_{\cD^{0}} = 0$ and $\Ric_{\cD^{0}}(\E_{+},\E_{-}) = 0$. These two equations do not depend on the choice of the divergence-free connection $\cD^{0}$ and they can be viewed as a system of algebraic equations for the maximal positive subspace $\E_{+} \subseteq \d$. 
\end{theorem}
\begin{proof}
We consider $\cD^{0}$ to be divergence-free and $\g$ unimodular to satisfy the assumptions of Theorem \ref{thm_dilaton}. This ensures that we find $\phi$ such that $\cD^{\sigma}$ can be used in Theorem \ref{thm_main}. This together with Theorem \ref{thm_Requivalences} shows that $(g,B,\phi)$ solve the equations of motion, if and only if $\Ric^{+}_{\cD^{0}} = 0$ and $\Ric_{\cD^{0}}(\E_{+},\E_{-}) = 0$. This is clearly a system of algebraic equations. Finally, we have already noted in the paragraph under Theorem \ref{thm_main} that the two involved quantities $\RS_{\cD^{0}}^{+}$ and $\Ric_{\cD^{0}}^{+-}$ depend on $\cD^{0}$ only through its divergence, which is assumed to be trivial. 
\end{proof}
The algebraic system for $\E_{+} \subseteq \d$ is in detail derived and written down in Appendix \ref{ap_eom}. In particular, see Theorem \ref{thm_EOMalgebraic} for the summary of the results. 

\begin{theorem}[\textbf{Poisson--Lie T-duality of effective actions}] \label{thm_PLTdual}
The algebraic system of equations for $\E_{+}$, by construction, does not depend on the choice of the subalgebra $\g$. 

In particular, given a solution of such a system, we may produce a solution of the equations of motion given by the low-energy effective action functional (\ref{eq_effaction}) on \textit{every} coset space $D / G$, i.e. corresponding to different choices of $G$.
\end{theorem}

This should be understood as a proof of consistence of Poisson--Lie T-duality on the level of $\sigma$-models with the corresponding low-energy theories (type II supergravity).
\begin{rem} \label{rem_indefinite}
We have constructed and proved everything using the positive subspace $\E_{+} \subseteq \d$. Consequently, the metric $g$ is always Riemannian. In order to allow for any (e.g. Lorentzian) signature, one has to relax the conditions on $\E_{+}$. Namely, one can consider any half-dimensional subspaces $\E_{+} \subseteq \d$ such that the restriction of $\<\cdot,\cdot\>_{\d}$ to $\E_{+}$ remains non-degenerate. This is possible as long as the involved objects make sense. Some maps may become singular or negative. For example, in formula (\ref{eq_dilatonfinal}) we may be forced to add absolute values to the arguments of the logarithms. One only has to be more careful to cope with such issues. We will comment more on this in the examples in the following section.
\end{rem}
\section{Some examples} \label{sec_examples}
Let us now examine some non-trivial examples of Manin pairs which can be explicitly integrated to a group pair. We will first look at the algebraic system of equations of motion obtained in Theorem \ref{thm_EOMalgebraic}. In particular, we will choose some non-trivial Lie quasi-bialgebra $(\g,\delta,\mu)$.
\subsection{Lie algebra $\g$ is Abelian}
First, note that an Abelian $\g$ is certainly unimodular. All terms containing the bracket $[\cdot,\cdot]_{\g}$ disappear. Moreover, the coboundary operator $\Delta$ in the Lie algebra complex $\frc^{\bullet}(\g, \Lambda^{2}(\g))$ acts trivially and thus $\delta' = \delta$, whence $[\cdot,\cdot]'_{\g^{\ast}} = [\cdot,\cdot]_{\g^{\ast}}$, c.f. Appendix \ref{ap_eom}. In fact, in this case $[\cdot,\cdot]_{\g^{\ast}}$ actually is a Lie bracket and $d_{\ast}$ defined by (\ref{eq_dastalmostdif}) is a true differential. Moreover, one has $\mu' = \mu - d_{\ast}(\theta_{0})$. 

Next, as we want $\RS_{\cD^{0}}^{+} = 0$ together with $\Ric_{\cD^{0}}(\E_{+},\E_{-}) = 0$, the equation (\ref{eq_EOMscalar2}) must hold. In this case, it reduces to
\begin{equation} \label{eq_muprimenaquadrat}
\<\mu',\mu'\>_{g_{0}} = 0. 
\end{equation}
If we assume that $g_{0}$ is positive definite, this equation forces $\mu' = 0$. In other words, we require $\theta_{0}$ to be the potential for the $3$-vector $\mu$, namely 
\begin{equation}
d_{\ast}(\theta_{0}) = \mu.
\end{equation}
This provides a first algebraic obstruction to the existence of some solution $(g_{0},\theta_{0})$, namely the Chevalley--Eilenberg cohomology class $[\mu] \in \mathfrak{H}^{3}(\g^{\ast}, \R)$ must be trivial. Note that this condition makes sense as $\mu$ is by construction (it is one of the Lie quasi-bialgebra conditions) $d_{\ast}$-closed. We can thus set $\mu' = 0$ in all the remaining equations. In particular, note that the skew-symmetric tensor (\ref{eq_EOMtensor2}) vanishes for any Lie algebra $\g^{\ast}$. Moreover, $\RS_{\cD^{+}}^{0}$ is proportional to the trace of $\Ric_{s}^{0}(x,y)$ and we are left with 
\begin{equation}
\begin{split}
0 = \Ric^{0}_{s}(x,y) = & \ - \frac{1}{4}\< [t^{i},t^{j}]_{\g^{\ast}}, x \> \cdot \<[e_{i},e_{j}]_{\g^{\ast}}, y \>\,+ \\ 
& + \frac{1}{2} g_{0}^{-1}([g_{0}(x),t^{i}]_{\g^{\ast}}, [g_{0}(y),e_{i}]_{\g^{\ast}}) \,+\\
& + \frac{1}{2} c_{\g^{\ast}}( g_{0}(x), g_{0}(y)).
\end{split}
\end{equation}
It is still quite difficult to solve this equation. So far, we have found a single non-trivial Lie algebra $\g^{\ast}$ where we were able to find the solution.

\begin{example}[\textbf{$\g^{\ast}$ is the Heisenberg algebra}] \label{ex_heisenberg1} Suppose $\g^{\ast}$ is a $3$-dimensional Lie algebra given in some basis $(t^{1},t^{2},t^{3})$ by the commutation relations 
\begin{equation} \label{eq_heisenberg}
[t^{1},t^{2}]_{\g^{\ast}} = 0, \; \; [t^{2},t^{3}]_{\g^{\ast}} = t^{1}, \; \; [t^{3},t^{1}]_{\g^{\ast}} = 0.
\end{equation}
This Lie algebra has a zero Killing form, $c_{\g^{\ast}} = 0$, and a non-trivial third cohomology group $\mathfrak{H}^{3}(\g,\R) = \R$. In fact, we have $d_{\ast}(\theta_{0}) = 0$ for all $\theta_{0} \in \Lambda^{2}(\g)$ and our only choice for $\mu$ is thus $\mu = 0$. The bivector $\theta_{0} \in \Lambda^{2} \g$ can be chosen arbitrarily. Note that in the following, we will allow for $g_{0}$ of any signature. In general, there could be thus some non-zero $\mu' \in \Lambda^{3}(\g)$ satisfying (\ref{eq_muprimenaquadrat}). However, in this case $\dim(\Lambda^{3}\g) = 1$ and there are no such elements. 

One can approach the problem as follows. The most general automorphism of $(\g^{\ast}, [\cdot,\cdot]_{\g^{\ast}})$ has the matrix in the above basis given by
\begin{equation}
\fA = \begin{pmatrix}
\det(\fM) & x^{T} \\
0 & \fM
\end{pmatrix},
\end{equation}
where $\fM \in \GL(2,\R)$ and $x \in \R^{2}$ are arbitrary. We can examine the action of this group on the space of symmetric non-degenerate matrices $\fg_{0}^{-1}$. In general, we have
\begin{equation}
\fg_{0}^{-1} = \bma{h_{0}}{a^{T}}{a}{\fh}, 
\end{equation}
where $h_{0} \in \R$, $\fh \in \Sym(2,\R)$ and $a \in \R^{2}$ take values constrained by the condition $\det(\fg_{0}^{-1}) \neq 0$. For $h_{0} \neq 0$, $\fg_{0}^{-1}$ lies in the orbit uniquely represented by one of the following symmetric non-degenerate matrices:
\begin{equation}
\fg_{0}^{-1}(\epsilon,\lambda_{1},\lambda_{2}) = \begin{pmatrix}
\epsilon & 0 & 0 \\
0 & \lambda_{1} & 0 \\
0 & 0 & \lambda_{2}
\end{pmatrix},
\end{equation}
where $\epsilon \in \{-1,1\}$ and $\lambda_{1} \geq \lambda_{2}$ are two non-zero numbers. The subset $h_{0} = 0$ contains exactly two orbits uniquely represented by
\begin{equation} \label{eq_fg0inveps}
\fg_{0}^{-1}(\epsilon) = \begin{pmatrix}
0 & 0 & 1 \\
0 & \epsilon & 0 \\
1 & 0 & 0 
\end{pmatrix},
\end{equation}
where $\epsilon \in \{-1,1\}$. Note that no metric in these two classes is positive definite, as $g_{0}^{-1}(t^{1},t^{1}) = 0$. It is easy to see that the condition $\Ric_{s}^{0} = 0$ has to be verified only for a single representative of each orbit. We can thus evaluate it on $g_{0}^{-1}$ defined in the basis $(t^{1},t^{2},t^{3})$ by matrices $\fg_{0}^{-1}(\epsilon,\lambda_{1},\lambda_{2})$ and $\fg_{0}^{-1}(\epsilon)$ above. In the first case, the matrix of $\Ric_{s}^{0}$ has the form
\begin{equation}
\Ric_{s}^{0} = \begin{pmatrix}
- \frac{1}{2 \lambda_{1} \lambda_{2}} & 0 & 0 \\
0 & - \frac{1}{2 \lambda_{1}^{2} \lambda_{2}} & 0 \\
0 & 0 & - \frac{1}{2 \lambda_{1} \lambda_{2}^{2}}
\end{pmatrix},
\end{equation}
which is non-zero. The family $\fg_{0}^{-1}(\epsilon,\lambda_{1},\lambda_{2})$ thus does not solve the equations. On the other hand, it is easy to see that $\fg_{0}^{-1}(\epsilon)$ both \textit{do solve} the equations. Note that $g_{0}$ has signature $(2,1)$ for $\epsilon = 1$ and $(1,2)$ for $\epsilon = -1$. There is no positive-definite solution.

In this example, everything can be calculated very explicitly. Indeed, note that $\g^{\ast}$ is the Heisenberg algebra realized by a subalgebra of matrices 
\begin{equation}
t^{1} = \begin{pmatrix}
0 & 0 & 1 \\
0 & 0 & 0 \\
0 & 0 & 0 
\end{pmatrix}, \; \; 
t^{2} = \begin{pmatrix}
0 & 1 & 0 \\
0 & 0 & 0 \\
0 & 0 & 0 
\end{pmatrix}, \; \; 
t^{3} = \begin{pmatrix}
0 & 0 & 0 \\
0 & 0 & 1 \\
0 & 0 & 0 
\end{pmatrix}.
\end{equation}
This Lie algebra integrates to the continuous Heisenberg group $\text{H}_{3}(\R)$; a closed subgroup of $\GL(3,\R)$ consisting of upper-triangular matrices
\begin{equation}
[\alpha_{1},\alpha_{2},\alpha_{3}] = \begin{pmatrix}
1 & \alpha^{2} & \alpha^{1} \\
0 & 1 & \alpha^{3} \\
0 & 0 & 1
\end{pmatrix},
\end{equation}
where $\alpha^{i} \in \R$ are arbitrary real numbers. As manifolds, we have $\text{H}_{3}(\R) \cong \R^{3}$. We may view $\alpha^{i}$ as global coordinate functions on $\text{H}_{3}(\R)$. The matrix of the adjoint representation of $\text{H}_{3}(\R)$ in the above basis of $\g^{\ast}$ is 
\begin{equation}
\Ad_{[\alpha^{1},\alpha^{2},\alpha^{3}]} = \begin{pmatrix}
1 & - \alpha^{3} & \alpha^{2} \\
0 & 1 & 0 \\
0 & 0 & 1 
\end{pmatrix}
\end{equation}
Now, as $\g$ is Abelian, we can choose $G = (\g,+)$. The Lie group $D$ can be then chosen as 
\begin{equation}
D = \g \rtimes \text{H}_{3}(\R),
\end{equation}
where $\text{H}(3)$ acts on $\g$ via the coadjoint representation. As both $G = \g$ and $\text{H}_{3}(\R)$ are subalgebras of $D$, we are in the Manin triple scenario described in \cite{Jurco:2017gii}. In particular, the coset space $S$ can be identified with $\text{H}_{3}(\R)$ and the projection $\pi: D \rightarrow S$ is in this case indeed just a projection
\begin{equation}
\pi(x,[\alpha^{1},\alpha^{2},\alpha^{3}]) = [\alpha^{1},\alpha^{2},\alpha^{3}]. 
\end{equation}
Moreover, we have $\d = \g \oplus \g^{\ast}$, the canonical inclusion $\tj: \g^{\ast} \rightarrow \g \oplus \g^{\ast}$ is everywhere admissible and for any $\xi \in \g^{\ast}$, we have $\xi^{\tr} = \xi^{R}$, a right-invariant vector field on $\text{H}_{3}(\R)$ generated by $\xi \in \g^{\ast} = \Lie(\text{H}_{3}(\R))$. Naturally, the dual generators $x^{\tr}$ are right-invariant $1$-forms $x_{R}$ generated by $x \in \g$. Explicitly, one has 
\begin{equation}
\xi^{R} = (\xi_{1} + \xi_{2} \alpha_{3}) \frac{\partial}{\partial \alpha^{1}} + \xi_{2} \frac{\partial}{\partial \alpha^{2}} + \xi_{3} \frac{\partial}{\partial \alpha^{3}},
\end{equation}
where $\xi_{i} = \xi(t_{i})$ for all $i \in \{1, 2, 3\}$. Next, we must find the map $\fPi \in C^{\infty}(\text{H}_{3}(\R), \Hom(\g,\g^{\ast}))$. By \ref{tvrz_fPiasck}, it may be found using the blocks of the adjoint representation. One has
\begin{equation}
\Ad_{(0,[\alpha^{1},\alpha^{2},\alpha^{3}])} = \bma{ \Ad^{\ast}_{[\alpha^{1},\alpha^{2},\alpha^{3}]}}{0}{0}{ \Ad_{[\alpha^{1},\alpha^{2},\alpha^{3}]}}.
\end{equation}
We see that $\fPi = 0$. This is hardly surprising as $\fPi$ is (up to a sign) the unique multiplicative Poisson bivector on $\text{H}_{3}(\R)$ whose Lie derivative at the unit should give a cocycle defining the trivial Lie algebra structure on $\g$. From (\ref{eq_fEasfE0Pi}), we thus get 
\begin{equation}
\fE = (E_{0}^{-1} - \fPi)^{-1} = E_{0} = g_{0}^{-1} + \theta_{0}. 
\end{equation}
This shows that the background field $B$ is the right-invariant $2$-form on $\text{H}_{3}(\R)$ generated by $\theta_{0} \in \R$. As $\mu = 0$, the formula (\ref{eq_Hinadmissiblesplit}) implies that $H = 0$. The metric $g$ is right-invariant and generated by $g_{0}^{-1} \in S^{2}(\g)$. It follows from the above discussion that $g_{0}^{-1}$ can be arbitrary metric on $\g^{\ast}$ such that $t^{1}$ is its isotropic vector, $g_{0}^{-1}(t^{1},t^{1}) = 0$. For example, if $g_{0}^{-1}$ has the matrix (\ref{eq_fg0inveps}), we find 
\begin{equation}
\begin{split}
g = & \ \epsilon \; {\rm d}\alpha^{2} \otimes {\rm d}\alpha^{2} + {\rm d}\alpha^{1} \otimes {\rm d}\alpha^{3} + {\rm d}\alpha^{3} \otimes {\rm d}\alpha^{1}\,- \\
& - \alpha_{3}( {\rm d}\alpha^{2} \otimes {\rm d}\alpha^{3} + {\rm d}\alpha^{3} \otimes {\rm d}\alpha^{2}). 
\end{split}
\end{equation}
Finally, we can plug into the formula (\ref{eq_dilatonfinal}). The Lie group $\text{H}_{3}(\R)$ is unimodular and thus $\bm{\nu} = 1$. We thus find a quite boring formula for dilaton, namely $\phi = 0$. 
\end{example}
\subsection{The class of $\delta$ is trivial}
Let $\g$ be any unimodular Lie algebra. Suppose we are given an isotropic splitting $\tj: \g^{\ast} \rightarrow \d$ and calculate the corresponding Chevalley--Eilenberg $1$-cocycle $\delta: \g \rightarrow \Lambda^{2} \g$. In fact, its class in the Chevalley--Eilenberg cohomology $[\delta]$ is independent the choice of $\tj$. This is in fact explained in the paragraph under the equation (\ref{eq_deltachangeofsplitting}). Suppose that this class is trivial, $[\delta] = 0$. 

This implies that we may choose the splitting $\tj: \g^{\ast} \rightarrow \d$, so that $\delta = 0$. Correspondingly, one has $[\cdot,\cdot]_{\g^{\ast}} = 0$. The bracket $[\cdot,\cdot]'_{\g^{\ast}}$ can however be highly non-trivial, as it is equal to
\begin{equation}
[\xi,\eta]'_{\g^{\ast}} = \ad^{\ast}_{\theta_{0}(\xi)}(\eta) - \ad^{\ast}_{\theta_{0}(\eta)}(\xi). 
\end{equation}
For example, one has $\fa' = [\theta_{0}(t^{i}),t_{i}]_{\g}$ which is in general non-zero. The Killing form $c'_{\g^{\ast}}$ gives
\begin{equation}
\begin{split}
c'_{\g^{\ast}}(\xi,\eta) = & \ c_{\g}( \theta_{0}(\xi), \theta_{0}(\eta)) - \< \xi, [\theta_{0}(t^{i}), [\theta_{0}(\eta), t_{i}]_{\g}]_{\g} \> \,- \\
& - \< \eta, [\theta_{0}(t^{i}), [\theta_{0}(\xi), t_{i}]_{\g}]_{\g} \> \,+\\
& + \< \xi, [t_{i}, \theta_{0}(t^{j})]_{\g} \> \cdot \< \eta, [t_{j}, \theta_{0}(t^{i})]_{\g} \>.
\end{split}
\end{equation}
Nevertheless, there are examples where we can solve the equations of motion completely. 
\begin{example}[\textbf{$\g$ is now the Heisenberg algebra}]
Suppose $\g$ is a $3$-dimensional Lie algebra given in the basis $(t_{1},t_{2},t_{3})$ by its commutation relations
\begin{equation}
[t_{1},t_{2}]_{\g} = 0, \; \; [t_{2},t_{3}]_{\g} = t_{1}, \; \; [t_{3},t_{1}]_{\g} = 0.
\end{equation}
Using arguments parallel to previous subsection, it is sufficient to verify the equations of motion for $g_{0}$ having the matrix, in this basis, in one of the following forms:
\begin{equation}
\fg_{0}(\epsilon,\lambda_{1},\lambda_{2}) = \begin{pmatrix}
\epsilon & 0 & 0 \\
0 & \lambda_{1} & 0 \\
0 & 0 & \lambda_{2}
\end{pmatrix}, \; \; 
\fg_{0}(\epsilon) = \begin{pmatrix}
0 & 0 & 1 \\
0 & \epsilon & 0 \\
1 & 0 & 0 
\end{pmatrix},
\end{equation}
where $\epsilon \in \{-1,1\}$ and $\lambda_{1} \geq \lambda_{2}$ are arbitrary non-zero numbers. Moreover, let us parametrize the matrix $\ftheta_{0}$ of the bivector $\theta_{0}$ in the above basis as
\begin{equation} \label{eq_theta0param}
\ftheta_{0}(b_{1},b_{2},b_{3}) = \begin{pmatrix}
0 & b_{3} & -b_{2} \\
-b_{3} & 0 & b_{1} \\
b_{2} & -b_{1} & 0 
\end{pmatrix}.
\end{equation}
In other words, we have $\ftheta_{0}^{ij}(b_{1},b_{2},b_{3}) = \epsilon^{ijk} b_{k}$ for a triple of unknown constants $(b_{1},b_{2},b_{3})$. The best approach is calculate first the bracket $[\cdot,\cdot]'_{\g^{\ast}}$ and all the related quantities. One finds that the commutation relations in the dual basis $(t^{1},t^{2},t^{3})$ of $\g^{\ast}$ are 
\begin{equation}
[t^{1},t^{2}]'_{\g^{\ast}} = b_{1}t^{2}, \; \; [t^{2},t^{3}]'_{\g^{\ast}} = 0, \; \; [t^{3},t^{1}]'_{\g^{\ast}} = - b_{1}t^{3}. 
\end{equation}
Note that in this basis, the bracket depends on $\theta_{0}$ only through the parameter $b_{1}$. For the induced quantities, we find their coordinate expressions
\begin{equation}
\fa' = \begin{pmatrix}
2b_{1} \\
0 \\
0
\end{pmatrix}, \; \; c'_{\g^{\ast}} = \begin{pmatrix}
2b_{1}^{2} & 0 & 0 \\
0 & 0 & 0 \\
0 & 0 & 0 
\end{pmatrix}.
\end{equation}
In fact, for $b_{1} \neq 0$, the Lie algebra $(\g^{\ast}, [\cdot,\cdot]'_{\g^{\ast}})$ lies in the isomorphism class Bianchi $\mathbf{5}$ (see e.g. Appendix A of \cite{Snobl:2002kq}) of the standard Bianchi classification. Suppose $\mu$ is in the above basis written as $\mu = \mu_{0} \cdot t_{1} \^ t_{2} \^ t_{3}$ for some given constant $\mu_{0} \in \R$. Then 
\begin{equation}
\mu' = (\mu_{0} + b_{1}^{2}) \cdot t_{1} \^ t_{2} \^ t_{3}. 
\end{equation}
This shows that the entire system of algebraic equations in Theorem \ref{thm_EOMalgebraic} does depend (in this basis) on $\theta_{0}$ only through $b_{1}$. Now, we can treat the two above classes of metrics separately.
\begin{enumerate}[i)]
\item $\fg_{0}(\epsilon,\lambda_{1},\lambda_{2})$: Let $\nu_{0} = \mu_{0} + b_{1}^{2}$. The best way to start is to solve the two scalar equations (\ref{eq_EOMscalar2}, \ref{eq_EOMscalar3}). This ensures that the scalar equation (\ref{eq_EOMscalar1}) holds and it renders $\Ric_{s}^{0}$ traceless. We obtain two equations 
\begin{equation}
\nu_{0}^{2} \lambda_{1}\lambda_{2} - \frac{1}{\lambda_{1}\lambda_{2}} = 0, \; \; \nu_{0} - \nu_{0}^{2} \lambda_{1}\lambda_{2} - 2 b_{1}^{2} = 0. 
\end{equation}
The first of the two conditions forces $\nu_{0} = \pm \frac{1}{\lambda_{1}\lambda_{2}}$. Plugging this into the second equation, we find the condition on $b_{1}$ in the form
\begin{equation}
b_{1}^{2} = \frac{\pm 1 - 1}{2\lambda_{1}\lambda_{2}}.
\end{equation}
One thus has the following two possibilities. 
\begin{enumerate}[(a)]
\item $\nu_{0} = \frac{1}{\lambda_{1}\lambda_{2}}$: This implies $b_{1} = 0$. In particular, $\nu_{0} = \mu_{0}$. We have $[\cdot,\cdot]'_{\g^{\ast}} = 0$. Consequently, $\Ric_{a}^{0} = 0$. It is straightforward to see that also $\Ric_{s}^{0} = 0$. We have thus found a solution. Suppose $\mu_{0} \neq 0$ is arbitrary. We can write $\lambda_{2} = \frac{1}{\lambda_{1}\mu_{0}}$. For a given $\mu_{0} \neq 0$ we thus have a solution $E_{0} = (g_{0}^{-1} + \theta_{0})^{-1}$. Writing the matrices in the above basis, we find
\begin{equation}
\begin{split}
\fg_{0}(\epsilon,\lambda_{1}) = & \ \begin{pmatrix}
\epsilon & 0 & 0 \\
0 & \lambda_{1} & 0 \\
0 & 0 & \frac{1}{\lambda_{1} \mu_{0}}
\end{pmatrix}, \\
\ftheta_{0}(b_{2},b_{3}) = & \ \begin{pmatrix}
0 & b_{3} & -b_{2} \\
-b_{3} & 0 & 0 \\
b_{2} & 0 & 0
\end{pmatrix},
\end{split}
\end{equation}
where $\epsilon \in \{-1,1\}$, $\lambda_{1} \neq 0$ and $b_{2},b_{3} \in \R$ are arbitrary parameters. 
\item $\nu_{0} = - \frac{1}{\lambda_{1}\lambda_{2}}$: This forces $b_{1}^{2} = - \frac{1}{\lambda_{1}\lambda_{2}}$. In particular, this can work only for $\lambda_{1}\lambda_{2} < 0$. We also see that $\mu_{0} = 0$. One can immediately argue that we have a solution. Indeed, this case corresponds to $\mu_{0} < 0$ in the (a) case, where we twist the splitting to get $\mu_{0} = 0$ at the expense of non-zero $b_{1}$. However, one can also verify everything by an explicit calculation. We indeed obtain a solution whose matrices are
\begin{equation}
\begin{split}
\fg_{0}(\epsilon,\lambda_{1},\lambda_{2}) = & \ \begin{pmatrix}
\epsilon & 0 & 0 \\
0 & \lambda_{1} & 0 \\
0 & 0 & \lambda_{2}
\end{pmatrix}, \\
\ftheta_{0}(\lambda_{1},\lambda_{2},b_{2},b_{3}) = & \ \begin{pmatrix}
0 & b_{3} & -b_{2} \\
-b_{3} & 0 & \frac{1}{\sqrt{-\lambda_{1}\lambda_{2}}} \\
b_{2} & \frac{-1}{\sqrt{-\lambda_{1}\lambda_{2}}} & 0
\end{pmatrix}.
\end{split}
\end{equation}
The parameters $\epsilon \in \{-1,1\}$, $\lambda_{1},\lambda_{2} \neq 0$ and $b_{2},b_{3} \in \R$ are arbitrary except of the condition $\lambda_{1}\lambda_{2} < 0 $. Here $\ftheta_{0}$ is not independent of the choice of $\fg_{0}$. 
\end{enumerate}
\item $\fg_{0}(\epsilon)$: The scalar equation (\ref{eq_EOMscalar3}) in this case immediately implies $\mu' = 0$, that is $\mu_{0} + b_{1}^{2} = 0$. We thus have to assume $\mu_{0} \leq 0$. The other scalar equation (\ref{eq_EOMscalar1}) gives no restrictions. One can show that $\Ric_{a}^{0} = 0$ and $\Ric_{s}^{0}$ has the matrix
\begin{equation}
\Ric_{s}^{0} = \begin{pmatrix}
0 & 0 & 0 \\
0 & 0 & 0 \\
0 & 0 & 2 b_{1}^{2}
\end{pmatrix}.
\end{equation}
In other words, the only possibility is $b_{1} = 0$. We obtain a class of solutions
\begin{equation}
\fg_{0}(\epsilon) = \begin{pmatrix}
0 & 0 & 1 \\
0 & \epsilon & 0 \\
1 & 0 & 0
\end{pmatrix}, \; \; \ftheta_{0}(b_{2},b_{3}) = \begin{pmatrix}
0 & b_{3} & -b_{2} \\
-b_{3} & 0 & 0 \\
b_{2} & 0 & 0
\end{pmatrix},
\end{equation}
where $\epsilon \in \{-1,1\}$ and $b_{2},b_{3} \in \R$ are arbitrary parameters. 
\end{enumerate}
There are some remarks in order. For $\mu_{0} = 0$, the Lie algebra $\d$ is isomorphic to the one in Example \ref{ex_heisenberg1}. All the corresponding solutions in (i)-(b) and (ii) can be thus obtained via Theorem \ref{thm_PLTdual} from those in Example \ref{ex_heisenberg1}. For $\mu_{0} < 0$, we can choose another isotropic splitting $\tj': \g^{\ast} \rightarrow \d$ to make $\mu'_{0} = 0$. The Manin pair $(\d,\g)$ thus allows a decomposition into the Manin triple $(\d,\g,\g^{\ast})$. This leads precisely to the Manin triple isomorphic to the dual to $(\mathbf{5}|\mathbf{2.i})$ listed in Appendix $B$ of \cite{Snobl:2002kq}. But in fact, this Manin triple (see Theorem $1$ of the same reference) is the decomposition of the same Drinfel'd double as the semi-Abelian case $(\mathbf{5}|\mathbf{1})$. The only solution that stands out is thus $\mu_{0} > 0$ case in (i)-(a). Note that for $\epsilon = 1$ and $\lambda_{1} > 0$ the metric $g_{0}$ is positive-definite. This indirectly proves that $\d$ is not isomorphic to the semi-Abelian case with $\mu_{0} = 0$, as there is no positive-definite solution. 

For $\mu_{0} = 0$, we can explicitly integrate $(\d,\g)$ to $(D,G)$. Indeed, we can set
\begin{equation}
D = H_{3}(\R) \ltimes \g^{\ast},
\end{equation}
where $\g = \Lie(H_{3}(\R))$ is the Lie algebra discussed in this example. Here $H_{3}(\R)$ acts on $\g^{\ast}$ by its coadjoint representation. The coset space $S$ can be identified with $\g^{\ast}$. The quotient map $\pi: H_{3}(\R) \ltimes \g^{\ast} \rightarrow \g^{\ast}$ is identified with the map defined for all $[\alpha^{1},\alpha^{2},\alpha^{3}] \in H_{3}(\R)$ and $\eta \in \g^{\ast}$ as 
\begin{equation}
\pi( [\alpha^{1},\alpha^{2},\alpha^{3}], \eta) = \Ad^{\ast}_{[\alpha^{1},\alpha^{2},\alpha^{3}]}(\eta).
\end{equation}
Here, we used the notation introduced in Example \ref{ex_heisenberg1}. One can show that the bivector $\Pi_{S}$ on $S \cong \g^{\ast}$ corresponding to the canonical splitting $\tj: \g^{\ast} \rightarrow \g \oplus \g^{\ast}$ is minus the Kirillov-Kostant-Souriau Poisson structure $\Pi_{\g^{\ast}}$ on $\g^{\ast}$. Recall that we can view the basis $(t_{k})_{k=1}^{\dim(\g)}$ as coordinates on $\g^{\ast}$. The bivector $\Pi_{\g^{\ast}}$ is then given by 
\begin{equation}
\Pi_{\g^{\ast}}(\xi) = \frac{1}{2}\<\xi, [t_{i},t_{j}]_{\g} \> \frac{\partial}{\partial t_{i}} \^ \frac{\partial}{\partial t_{j}}
\end{equation}
The right-invariant vector fields on the Abelian Lie group $(\g^{\ast},+)$ generated by $t_{i} \in \g^{\ast}$ coincide with the coordinate vector fields $\frac{\partial}{\partial t_{i}}$. In other words, the function $\fPi \in C^{\infty}(\g^{\ast}, \Lambda^{2} \g^{\ast})$ can be for all $x,y \in \g$ written as 
\begin{equation}
(\fPi(x,y))(\xi) = - \<\xi, [x,y]_{\g} \>. 
\end{equation}
Equivalently, its $\g^{\ast}$-dependent matrix in the basis $(t_{1},t_{2},$ $t_{3})$ evaluated at $\xi \in \g^{\ast}$ takes the form 
\begin{equation}
\fPi(\xi) = \begin{pmatrix}
0 & 0 & 0 \\
0 & 0 & -\xi_{1} \\
0 & \xi_{1} & 0 
\end{pmatrix}
\end{equation}
As an example, consider the solution in the part (i)-(b) above. The matrix $\fE_{0}$ of the map $E_{0}: \g^{\ast} \rightarrow \g$ is then
\begin{equation}
\fE_{0} = \begin{pmatrix}
\epsilon & b_{3} & -b_{2} \\
-b_{3} & \frac{1}{\lambda_{1}} & \frac{1}{\sqrt{-\lambda_{1}\lambda_{2}}} \\
b_{2} & - \frac{1}{\sqrt{-\lambda_{1}\lambda_{2}}} & \frac{1}{\lambda_{2}}
\end{pmatrix}.
\end{equation}
One has $\det(\fE_{0}) = \frac{b_{2}^{2}}{\lambda_{1}} + \frac{b_{3}^{2}}{\lambda_{2}}$. This means that $\fE_{0}$ can be singular for some values of $b_{2}$ and $b_{3}$. Recall that $\lambda_{1}$ and $\lambda_{2}$ have opposite signs, so $\det(\fE_{0}) = 0$ can happen also for $(b_{2},b_{3}) \neq (0,0)$. These kinds of singularities appear as $\fg_{0}(\epsilon,\lambda_{1},\lambda_{2})$ is now indefinite. We will not go into the calculation of $\fE$. Instead, let us look at the dilaton $\phi$. As $\g^{\ast}$ is a unimodular Lie algebra, one has $\bm{\nu} = 1$. Moreover, as $g_{0}$ is indefinite, we will add the absolute values into the formula (\ref{eq_dilatonfinal}), see Remark \ref{rem_indefinite}. One has 
\begin{equation}
\begin{split}
\phi(\xi) = & \ - \frac{1}{2} \ln| \det(\f1_{\g} - E_{0} \fPi(\xi))| \\
= & \ - \frac{1}{2}\ln|\det \begin{pmatrix} 1 & b_{2}\xi_{1} & b_{3}\xi_{1} \\
0 & 1 - \frac{\xi_{1}}{\sqrt{-\lambda_{1}\lambda_{2}}} & \frac{\xi_{1}}{\lambda_{1}} \\
0 & - \frac{\xi_{1}}{\lambda_{2}} & 1 - \frac{\xi_{1}}{\sqrt{-\lambda_{1}\lambda_{2}}} 
\end{pmatrix}|
\\
= & \ - \frac{1}{2} \ln| 1 - \frac{2\xi_{1}}{\sqrt{-\lambda_{1}\lambda_{2}}}|.
\end{split}
\end{equation}
We see that this formula has a singularity at $\xi_{1} =$\linebreak $\frac{1}{2} \sqrt{-\lambda_{1}\lambda_{2}}$, another consequence of the indefiniteness of $g_{0}$. This concludes this example. 
\end{example}

In general, it is quite difficult to find some solutions. In fact, there is a (majority) of cases, where there is no solution at all. Let us demonstrate this on the following example. 

\begin{example}[\textbf{$\g$ is the compact algebra $\su(2)$}] This case is potentially very interesting as it includes the case where the geometry is easy to handle. 

Indeed, let $G$ be any Lie group, such that its Lie algebra $\g$ is quadratic, that is it comes equipped with a non-degenerate symmetric invariant bilinear form $g_{\g} \equiv \<\cdot,\cdot\>_{\g}$. Consider the direct product of Lie groups $D = G \times G$. We can view $G$ as a closed subgroup of $D$ if we identify it with its image with respect to the diagonal map $\Delta_{G}: G \rightarrow G \times G$. The inclusion map $\ti: \g \rightarrow \d$ is then equal to the diagonal map $\Delta_{\g}: \g \rightarrow \g \oplus \g$. As $\Delta_{\g}(\g) \subseteq \d$ must be Lagrangian, we define
\begin{equation}
\<(x,x'),(y,y')\>_{\d} = \<x,y\>_{\g} - \<x',y'\>_{\g}. 
\end{equation}
Clearly $(\d,\g)$ then forms a Manin pair. Define the isotropic splitting $\tj: \g^{\ast} \rightarrow \d$ for all $\xi \in \g^{\ast}$ by 
\begin{equation}
\tj(\xi) = \frac{1}{2} ( g_{\g}^{-1}(\xi), - g_{\g}^{-1}(\xi)). 
\end{equation}
It is easy to see that $\delta: \g \rightarrow \Lambda^{2} \g$ corresponding to this splitting vanishes. The $3$-vector $\mu$ can be then calculated to give
\begin{equation}
\begin{split}
\mu(\xi,\eta,\zeta) = & \ \frac{1}{4} \< [g_{\g}^{-1}(\xi), g_{\g}^{-1}(\eta)]_{\g}, \zeta \> \\
= & \ \frac{1}{4} (\Lambda^{3} g_{\g}^{-1})(\chi_{\g}) (\xi,\eta,\zeta),
\end{split}
\end{equation} 
where $\chi_{\g}(x,y,z) = \<[x,y]_{\g},z\>_{\g}$ is the Cartan $3$-form corresponding to $(\g,\<\cdot,\cdot\>_{\g})$. We see that this example fits into this subsection. Now, it is natural to consider $\g = \su(2)$ and choose $\<\cdot,\cdot\>_{\g}$ to be the Killing form $c_{\g}$ of $\g$. We find the following statement:
\begin{tvrz}
For this Manin pair, there is no $\E_{+} \subseteq \d$ solving the equations in Theorem \ref{thm_EOMalgebraic}.  This is true even if we allow arbitrary $\mu$ and the metric $g_{0}$ of any signature. 
\end{tvrz}
\begin{proof}Let us only sketch the approach leading to this statement. Let $(t_{1},t_{2},t_{3})$ be the usual $\su(2)$ basis where the commutation relations have the form
\begin{equation}
[t_{1},t_{2}]_{\g} = t_{3}, \; \; [t_{2},t_{3}]_{\g} = t_{1}, \; \; [t_{3},t_{1}]_{\g} = t_{2}.
\end{equation}
The matrix $\fA$ of any automorphism of $\su(2)$ in the above basis lies in the group $\text{SO}(3)$. We can thus change the basis such that the above commutation relations remain valid and the matrix $\fg_{0}$ of $g_{0}$ is diagonal. This simplifies calculations significantly. Moreover, one can parametrize $\theta_{0}$ by a vector $\fb = (b_{1},b_{2},b_{3})^{T}$ as in (\ref{eq_theta0param}). One can then combine (\ref{eq_EOMscalar1}) and (\ref{eq_EOMtensor2}) to show that necessarily $\fb = 0$. The compatibility of the two scalar equations (\ref{eq_EOMscalar2}) and (\ref{eq_EOMscalar3}) then forces conditions on $g_{0}$ which directly contradict the equations coming from the equation (\ref{eq_EOMtensor1}). 
\end{proof}
\end{example}

\begin{appendices}
\section{Important relations} \label{ap_relations}
Let us spell out some relations between the objects defined in the paragraph above (\ref{eq_fPi}). We assume that $\tj$ is an everywhere admissible (otherwise the formulas work only locally) splitting of (\ref{eq_ManinSES}). Let us start with the commutator $[\xi^{\tr},\eta^{\tr}]$ of the vector fields generating $\X(S)$. For all $\xi,\eta \in \g^{\ast}$, one has
\begin{equation}
\begin{split}
[\xi^{\tr}, \eta^{\tr}] = & \ [\#^{\tr}(\tj(\xi)), \#^{\tr}(\tj(\eta))] = - \#^{\tr}([\tj(\xi), \tj(\eta)]_{\d}) \\
= & \ -\#^{\tr}( \ti(\mu(\xi,\eta,\cdot)) + \tj([\xi,\eta]_{\g^{\ast}}) ) \\
= & \ -[\xi,\eta]_{\g^{\ast}}^{\tr} + \fPi(\mu(\xi,\eta,\cdot))^{\tr}.
\end{split}
\end{equation} 
Analogously, one can calculate the further relations and we conclude that
\begin{subequations}
\begin{align}
\label{eq_genrel1} [\xi^{\tr}, \eta^{\tr}] = & \ -[\xi,\eta]^{\tr}_{\g^{\ast}} + \fPi( \mu(\xi,\eta,\cdot))^{\tr}, \\
\label{eq_genrel2} [\xi^{\tr}, \fPi(y^{\tr})] = & \ -(\ad^{\ast}_{y}(\xi))^{\tr} - \fPi(\ad^{\ast}_{\xi}(y))^{\tr}, \\
\label{eq_genrel3} [\fPi(x)^{\tr},\fPi(y)^{\tr}] = & \ \fPi([x,y]_{\g})^{\tr}.
\end{align}
\end{subequations}
The corresponding Lie derivatives are
\begin{subequations}
\begin{align}
\Li{\xi^{\tr}}( y^{\tr}) = & \  - (\ad^{\ast}_{\xi}(y))^{\tr} - \mu(\xi,\fPi(y),\cdot)^{\tr}, \\
\Li{\fPi(x)}( y^{\tr}) = & \  [x,y]^{\tr}_{\g} - (\ad^{\ast}_{\fPi(y)}(x))^{\tr}. 
\end{align}
\end{subequations}
Using those partial results, we can derive the identity
\begin{equation} \label{eq_Picrucial}
\begin{split}
\Li{\xi^{\tr}}( \fPi(x,y)) = & \ - \xi([x,y]_{\g}) - \fPi( \ad^{\ast}_{\xi}(x), y)\,- \\
& - \fPi(x, \ad^{\ast}_{\xi}(y)) - \mu(\xi,\fPi(x),\fPi(y)). 
\end{split}
\end{equation}
\section{Proof of the dilaton formula} \label{ap_dilaton}
\begin{proof}[The proof of Theorem \ref{thm_dilaton}]
First, note that the equation (\ref{eq_dilatonformula}) has in fact two independent components, which can be rewritten using the definition of the connection $\cD^{\sigma}$ as
\begin{subequations}
\begin{align}
\label{eq_dilatonequation1} \Li{X}(\phi) = & \ \frac{1}{2}( \Div_{\cD^{g}}(X) - \Div_{\cD^{E}}(\sigma(X))), \\
\label{eq_dilatonequation2} 0 = & \ \Div_{\cD^{E}}(\rho^{\ast}(\xi)).
\end{align}
\end{subequations}
The second equation certainly does not depend on $\sigma$. It follows from (\ref{eq_changeofsplitting}) that if (\ref{eq_dilatonequation2}) stands true, the right-hand side of (\ref{eq_dilatonequation1}) is also independent of the splitting. Moreover, it is clear that $\Div_{\cD^{E}}$ depends on $\cD^{0}$ only through its divergence operator which is assumed to be trivial. 

Next, as noted under the definition (\ref{eq_admissible}), there exists an isotropic splitting $\tj: \g^{\ast} \rightarrow \d$ of the sequence (\ref{eq_ManinSES}) admissible on a neighborhood $U$ of each point. We can thus locally define $\phi$ via the formula (\ref{eq_dilatonfinal}). First, let us note that it is well defined. We can write $\f1_{\g} - E_{0}\fPi = E_{0} \fE^{-1}$. As $E_{0}$ and $\fE$ have positive-definite symmetric parts, we have $\det( E_{0} \fE^{-1}) > 0$ on the whole $S$ and the first part of $\phi$ thus makes sense. Now, if $U$ is connected and contains $s_{0} = \pi_{0}(G)$, we have $\bm{\nu}(s_{0}) = 1$ and $\bm{\nu}(s) \neq 0$ for all $s \in U$. Whence, $\bm{\nu}(U) \subseteq \R^{+}$ and the second part of $\phi$ is well-defined on $U$. Now, suppose that $\tj': \g^{\ast} \rightarrow \d$ is an isotropic splitting of (\ref{eq_ManinSES}) admissible on the neighborhood $U'$, such that $U \cap U' \neq \emptyset$. There is thus a unique bivector $\theta \in \Lambda^{2} \g$ satisfying (\ref{eq_splitchange}). As both $\tj$ and $\tj'$ are admissible an $U \cap U'$, it follows that the map $\f1_{\g^{\ast}} - \fPi \theta$ is invertible on $U \cap U$', and the maps $\fPi'$ and $E'_{0}$ associated to $\tj'$ can be written as 
\begin{equation} \label{eq_splitchangedilaton}
\fPi' =  \fPi(\f1_{\g} - \theta \fPi)^{-1},\; \; E'_{0} = E_{0} - \theta.
\end{equation}
We thus get the expression valid on $U \cap U'$:
\begin{equation}
\begin{split}
\f1_{\g} - E'_{0} \fPi' = & \ \f1_{\g} - (E_{0} - \theta) \fPi (\f1_{\g} - \theta \fPi)^{-1} \\
= & \ ( \f1_{\g} - \theta \fPi - (E_{0} - \theta) \fPi))(\f1_{\g} - \theta \fPi)^{-1} \\
= & \ (\f1_{\g} - E_{0} \fPi)(\f1_{\g} - \theta \fPi)^{-1}.
\end{split}
\end{equation}
In particular, this implies $\det( \f1_{\g} - \theta \fPi) > 0$ on $U \cap U'$ and we have the relation
\begin{equation}
\begin{split}
- \frac{1}{2} \ln( \det(\f1_{\g} - E'_{0} \fPi')) = & \ - \frac{1}{2} \ln( \det(\f1_{\g} - E_{0} \fPi)) \\
& + \frac{1}{2} \ln( \det(\f1_{\g} - \theta \fPi)). 
\end{split}
\end{equation}
Next, for all $d \in \pi_{0}^{-1}(U \cap U')$, one finds the relation $\fk'(d) = (\f1_{\g} - \theta \fc(d) \fk(d)^{-1}) \cdot \fk(d)$. Using Proposition \ref{tvrz_fPiasck}, we can now on $U \cap U'$ express the function $\bm{\nu'}$ as 
\begin{equation}
\bm{\nu'} = \det( \f1_{\g} - \theta \fPi) \cdot \bm{\nu}.
\end{equation}
Note that this also implies that $\bm{\nu'}(U \cap U') \subseteq \R^{+}$ and if $U'$ is connected, then $\bm{\nu'}(U') \subseteq \R^{+}$. As $S$ is by definition connected (because $D$ is connected), we may use this to show that $\bm{\nu}$ in the formula (\ref{eq_dilatonfinal})  defined using a splitting admissible on any connected open set $U \subseteq S$ satisfies $\bm{\nu}(U) \subseteq \R^{+}$ and the second term in (\ref{eq_dilatonfinal}) is well-defined. Moreover, on $U \cap U'$, we find
\begin{equation}
- \frac{1}{2} \ln( \bm{\nu'}) = - \frac{1}{2} \ln( \bm{\nu}) - \frac{1}{2} \ln( \det(\f1_{\g} - \theta \fPi)),
\end{equation}
which shows that the two formulas for $\phi$ coincide on the intersection $U \cap U'$. We can thus use the local expressions to construct $\phi$ globally. 

In particular, we can without the loss of generality assume in the remainder of the proof that $\tj: \g^{\ast} \rightarrow \d$ is an everywhere admissible isotropic splitting of (\ref{eq_ManinSES}). Let us now show that under the assumptions of the theorem, the equation (\ref{eq_dilatonequation2}) holds. Let $(t_{\mu})_{\mu=1}^{\dim(\d)}$ be a fixed basis of $\d$. It suffices to prove it for $\xi = df$ and $f \in C^{\infty}(S)$. One finds 
\begin{equation} \label{eq_ficEonrhoastdf}
\Div_{\cD^{E}}( \rho^{\ast}(df)) = \Li{\#^{\tr}(t_{\mu})} \Li{\#^{\tr}(t^{\mu}_{\d})}(f) + \Div_{\cD^{0}}( \rho^{\ast}(df)),
\end{equation}
where $(t^{\mu}_{\d})_{\mu=1}^{\dim(\d)}$ is the basis of $\d$ satisfying $\< t_{\mu}, t^{\nu}_{\d} \>_{\d} = \delta_{\mu}^{\nu}$. In fact, one can show from the properties of $E$ that the left-hand side of (\ref{eq_ficEonrhoastdf}) viewed as an operator on $C^{\infty}(S)$ is a vector field. In particular, there is a unique vector field $Y \in \X(S)$, such that $\Li{Y}(f) = \Li{\#^{\tr}(t_{\mu})} \Li{\#^{\tr}(t^{\mu}_{\d})}(f)$. It follows that $Y$ represents the quadratic Casimir $t_{\mu}  t^{\mu}_{\d} \in \mathfrak{U}(\d)$, hence it must commute with all the generators $\#^{\tr}(x)$. This implies that $Y$ is invariant with respect to the transitive dressing action $\tr$ and it is determined by its value at the point $s_{0} = \pi_{0}(G)$. One can write
\begin{equation}
Y = -\Li{t^{k \tr}}( \fPi(t_{q},t_{k})) \cdot t^{q \tr},
\end{equation}
where $(t_{k})_{k=1}^{\dim(\g)}$ is an arbitrary basis of $\g$. This expression can be evaluated using (\ref{eq_Picrucial}). In particular, at $s_{0}$ there is $\fPi_{s_{0}} = 0$ and one obtains 
\begin{equation}
\begin{split}
Y_{s_{0}} = & \ - t^{k}( [t_{q},t_{k}]_{\g}) \cdot \#^{\tr}_{s_{0}}( \tj(t^{q})) \\
= & \ - \#^{\tr}_{s_{0}}( \tj( \Tr( \ad_{t_{q}}) \cdot t^{q})) = - \#^{\tr}_{s_{0}}( \tj(\falpha)),
\end{split}
\end{equation}
where $\falpha \in \g^{\ast}$ is defined by $\falpha(x) = \Tr( \ad_{x})$ for all $x \in \g$. At any $s \in S$ and $\pi_{0}(d) = s$, one has 
\begin{equation}
Y_{s} = - \#^{\tr}_{s}( \Ad_{d}( \tj(\falpha))).
\end{equation}
But our assumptions were $\falpha = 0$ (unimodularity of $\g$) and $\Div_{\cD^{0}} = 0$. It follows that the right-hand side of (\ref{eq_ficEonrhoastdf}) vanishes and (\ref{eq_dilatonequation2}) indeed holds. 

Now, assume that $\sigma: TS \rightarrow E$ is the splitting (\ref{eq_sigmaasxitr}) corresponding to the everywhere admissible splitting $\tj$. For $X = \xi^{\tr}$, one finds that the second term on the right-hand side of (\ref{eq_dilatonequation1}) vanishes:
\begin{equation}
\Div_{\cD^{E}}( \sigma(\xi^{\tr})) = \Div_{\cD^{E}}( \tj(\xi)) = \Div_{\cD^{0}}( \tj(\xi)) = 0.
\end{equation}
We thus obtain the simplified equation for $\phi$, namely, for all $\xi \in \g^{\ast}$, one has
\begin{equation} \label{eq_dilatonequation3}
\Li{\xi^{\tr}}(\phi) = \frac{1}{2} \Div_{\cD^{g}}( \xi^{\tr}).
\end{equation}

Let us proceed by calculating the differential of the function $\phi$ given by (\ref{eq_dilatonfinal}). First, denote the first of the two summands as $\phi_{0} = - \frac{1}{2} \ln(\det(\f1_{\g} - E_{0} \fPi))$ . Using the standard formulas for differentiation of logarithms and determinants, one arrives to the formula
\begin{equation}
d\phi_{0} = \frac{1}{2} \Tr( \fE \cdot d \fPi).
\end{equation}
Write $\fE = \fg + \fB$ for the decomposition of $\fE$ into its symmetric and skew-symmetric part. As $\fPi$ is skew-symmetric, only $\fB$ contributes to the sum above. Whence we get
\begin{equation} \label{eq_dphi0}
(d\phi_{0})(\xi) = \frac{1}{2} \Tr( \fB \cdot \Li{\xi^{\tr}}(\fPi)).
\end{equation}
For the second summand, let $\phi_{1} = -\frac{1}{2} \ln( \bm{\nu})$. Now, recall that $\xi^{\tr}$ is $\pi_{0}$-related to the right-invariant vector field $\tj(\xi)^{R} \in \X(D)$. Whence,
\begin{equation}
\begin{split}
(d\phi_{1})(\xi^{\tr}) \circ \pi_{0} = & \ -\frac{1}{2} \Li{\tj(\xi)^{R}}(\ln(\det(\fk))) \\
= & \ -\frac{1}{2} \Tr( \fk^{-1} \Li{\tj(\xi)^{R}}(\fk)) \\
= & \ -\frac{1}{2} \Tr\{ \ad^{\ast}_{\xi} + \mu(\xi, (\fc \cdot \fk^{-1})(\star), \cdot) \},
\end{split}
\end{equation}
where $\star$ denotes the input of the map whose trace we are taking. Now, we have to apply Proposition \ref{tvrz_fPiasck} to see that $\fc \cdot \fk^{-1} = \fPi \circ \pi_{0}$. Moreover, define $\fa \in \g$ for all $\xi \in \g^{\ast}$ by $\xi(\fa) = \Tr(\ad^{\ast}_{\xi})$. We can thus rewrite the resulting expression as 
\begin{equation} \label{eq_dphi1}
(d\phi_{1})(\xi^{\tr}) = \frac{1}{2} \< \xi,  \fa - \mu(\fPi(t_{k}), t^{k}, \cdot) \>. 
\end{equation}
To proceed, let us derive the way to differentiate the function $\fg$ along the vector fields. First, decompose the constant map $E_{0}^{-1} = g'_{0} + B'_{0}$ into its symmetric and skew-symmetric part. Then $\fE = (g'_{0} + (B'_{0} - \fPi))^{-1}$. Using the standard formulas for the symmetric and skew-symmetric parts of the inverse, we find
\begin{equation}
\begin{split}
\fg = & \ (g'_{0} - (B'_{0} - \fPi) g'^{-1}_{0} (B'_{0} - \fPi))^{-1}, \\
\fB = & \ -g'^{-1}_{0}  (B'_{0} - \fPi) \fg.
\end{split}
\end{equation}
Using that $g'_{0}$ and $B'_{0}$ are constant, it is not difficult to arrive to the equation
\begin{equation} \label{eq_Liofginverse}
\Li{X}(\fg^{-1}) = - \Li{X} (\fPi) \cdot \fB \fg^{-1} - \fg^{-1} \fB \cdot \Li{X}(\fPi),
\end{equation}
for all $X \in \X(S)$. From the definition of Levi-Civita connection and divergence, one finds directly
\begin{equation}
\frac{1}{2}\Div_{\cD^{g}}(\xi^{\tr}) = \frac{1}{2}\< [t^{k \tr}, \xi^{\tr}], t^{\tr}_{k} \> - \frac{1}{4} \Tr( \Li{\xi^{\tr}}(\fg^{-1}) \cdot \fg)
\end{equation}
Using (\ref{eq_genrel1}), the first term can be rewritten as
\begin{equation}
\frac{1}{2}\< [t^{k\tr}, \xi^{\tr}], t^{\tr}_{k} \> = \frac{1}{2} \< \xi, \fa - \mu( \fPi(t_{k}), t^{k},\cdot) \>. 
\end{equation} 
But this is exactly the contribution (\ref{eq_dphi1}) of $\phi_{1}$. On the other hand, the contribution of the second term can be evaluated using (\ref{eq_Liofginverse}) and the usual properties of trace. One gets
\begin{equation}
- \frac{1}{4} \Tr( \Li{\xi^{\tr}}( \fg^{-1}) \cdot \fg) = \frac{1}{2} \Tr( \fB \cdot \Li{\xi^{\tr}}(\fPi) ).
\end{equation}
But this is precisely the contribution (\ref{eq_dphi0}). We conclude that the equation (\ref{eq_dilatonequation3}) holds true and we have thus finished the proof.
\end{proof}
\section{Deriving the equations of motion} \label{ap_eom}
\subsection{General expressions}
Let $\cD^{0} \in \LC(\d,\E_{+})$ be any divergence-free Levi-Civita connection on the Courant algebroid $(\d,0,\<\cdot,\cdot\>_{\d},-[\cdot,\cdot]_{\d})$. We write down the system of equations equivalent to $\RS_{\cD^{0}}^{+} = 0$ and $\Ric_{\cD^{0}}(\E_{+},\E_{-}) = 0$ explicitly . First, let us define a tensor $\fk \in (\d^{\ast})^{\otimes 3}$ using $\cD^{0}$ as 
\begin{equation}
\fk(x,y,z) = \< \cD^{0}_{x}(y), z \>_{\d}, 
\end{equation}
for all $x,y,z \in \d$. The condition $\cD^{0}(g_{\d}) = 0$ is equivalent to $\fk \in \d^{\ast}  \otimes \Lambda^{2} \d^{\ast}$. The compatibility with $\E_{+}$ then implies $\fk(x,y_{+},z_{-}) = 0$ for all $x,y,z \in \d$, where $x_{\pm}$ denotes the projections onto $\E_{\pm}$. Next, let $\chi(x,y,z) = \<[x,y]_{\d}, z \>_{\d}$ be the canonical Cartan $3$-form on $(\d,[\cdot,\cdot]_{\d})$. It follows that $\cD^{0}$ is torsion-free if and only if $\fk$ satisfies
\begin{equation} \label{eq_kpluscyclic}
\fk(x,y,z) + \cyc(x,y,z) = - \chi(x,y,z),
\end{equation}
or equivalently $\fk_{a} = - \frac{1}{3} \chi$, where $\fk_{a} \in \Lambda^{3} \d^{\ast}$ is the complete skew-symmetrization of $\fk$. Finally, we have assumed that $\cD^{0}$ is divergence-free. Let $(t_{\mu})_{\mu=1}^{\dim(\d)}$ be an arbitrary basis of $\d$, and let $(t^{\mu}_{\d})_{\mu=1}^{\dim(\d)}$ be the basis defined by $\< t^{\mu}_{\d}, t_{\nu} \>_{\d} = \delta^{\mu}_{\nu}$. We find the expression
\begin{equation}
0 = \Div_{\cD^{0}}(z) = \fk(t_{\mu},z,t^{\mu}_{\d}) \equiv - \fk'(z),
\end{equation}
for all $z \in \d$. We may now proceed to the calculation of the generalized Ricci tensor $\Ric_{\cD^{0}} \in S^{2}(\d^{\ast})$. From the definition, we obtain the generalized Riemann tensor 
\begin{equation}
\begin{split}
R_{\cD^{0}}(w,z,x,y) = & \ \frac{1}{2} \{ \fk(x,g_{\d}^{-1}\fk(y,z,\cdot),w)\,- \\
& - \fk(y, g_{\d}^{-1}\fk(x,z,\cdot),w)\,+ \\
& + \fk(z, g_{\d}^{-1} \fk(w,x,\cdot),y)\,- \\
& - \fk(w, g_{\d}^{-1}\fk(z,x,\cdot),y)\,+ \\
& + \fk(g_{\d}^{-1}\fk(\cdot,x,y),z,w)\,+ \\
& + \fk([x,y]_{\d},z,w) + \fk([z,w]_{\d},x,y) \}.
\end{split}
\end{equation}
It follows that the generalized Ricci tensor reads
\begin{equation}
\begin{split}
\Ric_{\cD^{0}}(x,y) = & \ \frac{1}{2}\{ -\fk'(t_{\mu}) ( \fk(x,y,t^{\mu}_{\d}) + \fk(y,x,t^{\mu}_{\d}) )\,+ \\
& + \fk( g_{\d}^{-1} \fk(\cdot,t_{\mu},y),x,t^{\mu}_{\d})\,- \\
& - \fk(y, g_{\d}^{-1} \fk(t_{\mu},x,\cdot),t^{\mu}_{\d})\,- \\
& - \fk(t^{\mu}_{\d}, g_{\d}^{-1}\fk(x,t_{\mu},\cdot),y)\,+ \\
& + \fk([t_{\mu},x]_{\d},y,t^{\mu}_{\d}) + \fk([t_{\mu},y]_{\d}, x, t^{\mu}_{\d}) \}.
\end{split}
\end{equation}
The first term proportional to $\fk'(t_{\mu})$ vanishes. The remaining five terms can be recast using (\ref{eq_kpluscyclic}) to give a relatively simple expression 
\begin{equation} \label{eq_Ric0eqfinal}
\Ric_{\cD^{0}}(x,y) = \frac{1}{2} \fk(t_{\lambda},t_{\mu},y) ( \fk(t^{\lambda}_{\d},t^{\mu}_{\d},x) - 2 \fk(t^{\mu}_{\d}, t^{\lambda}_{\d}, x) ).
\end{equation}
For the record, we may now easily calculate the canonical Ricci scalar $\RS_{\cD^{0}}$. One finds
\begin{equation} \begin{split}
\RS_{\cD^{0}} \equiv & \ \Tr_{g_{\d}}(\Ric_{\cD^{0}}) \\
= & \ \frac{1}{2}  \fk(t_{\lambda},t_{\mu},t_{\nu}) ( \fk(t^{\lambda}_{\d}, t^{\mu}_{\d}, t^{\nu}_{\d}) - 2 \fk( t^{\mu}_{\d}, t^{\lambda}_{\d}, t^{\nu}_{\d}) ) \\
= & \ - \frac{1}{2} \fk( t_{\lambda}, t_{\mu}, t_{\nu}) ( \fk(t^{\nu}_{\d}, t^{\lambda}_{\d}, t^{\mu}_{\d})\,+ \\
&\kern.5cm + \fk(t^{\mu}_{\d}, t^{\lambda}_{\d}, t^{\nu}_{\d}) + \chi(t^{\lambda}_{\d}, t^{\mu}_{\d}, t^{\nu}_{\d})) \\
= & \ - \frac{1}{2}\fk(t_{\lambda}, t_{\mu},t_{\nu}) \chi(t^{\lambda}_{\d}, t^{\mu}_{\d}, t^{\nu}_{\d}) \\
= & \ \frac{1}{6} \chi(t_{\lambda},t_{\mu},t_{\nu}) \chi(t^{\lambda}_{\d}, t^{\mu}_{\d}, t^{\nu}_{\d}) \equiv  \< \chi, \chi \>_{\d}. 
\end{split}
\end{equation}
In other words, $\RS_{\cD^{0}}$ is the square of the norm of the canonical Cartan $3$-form $\chi$ with respect to the metric $\<\cdot,\cdot\>_{\d}$. We will show in the moment that for unimodular $\g$ it actually vanishes. This is in fact necessary for $\cD^{\sigma}$ to satisfy the assumption of theorem \ref{thm_main}. Note that $\<\chi,\chi\>_{\d}$ does not depend on the generalized metric $\E_{+} \subset \d$ and its vanishing is ensured solely by the algebraic structure of $(\d,\g)$. 

Now, let $\gm$ denote the positive definite metric on $\d$ induced by the generalized metric $\E_{+}$ as above (\ref{eq_pluscurvature}), so that $\RS_{\cD^{0}}^{+} = \Tr_{\gm}(\Ric_{\cD^{0}})$. Let $(t^{\mu}_{\gm})_{\mu=1}^{\dim(\d)}$ be defined by $\gm( t^{\mu}_{\gm}, t_{\nu}) = \delta^{\mu}_{\nu}$. We thus find the expression
\begin{equation}
\RS_{\cD^{0}}^{+} = \frac{1}{2} \fk(t_{\lambda},t_{\mu}, t_{\nu}) ( \fk(t^{\lambda}_{\d}, t^{\mu}_{\d}, t^{\nu}_{\gm}) - 2 \fk( t^{\mu}_{\d}, t^{\lambda}_{\d}, t^{\nu}_{\gm})).
\end{equation}
At this moment, the calculation becomes unpleasant. Recall that $\d = \E_{+} \oplus \E_{-}$. Choose the basis $(t_{\mu})_{\mu=1}^{\dim(\d)}$ to be adapted to this splitting, that is fix the two bases $(t_{k}^{+})_{k=1}^{\dim{(\g})}$ and  $(t_{k}^{-})_{k=1}^{\dim(\g)}$ of $\E_{+}$ and $\E_{-}$, respectively, and combine them into a basis of $\d$. 

There are again induced bases $(t^{k\pm}_{\gm})_{k=1}^{\dim(\g)}$ of $\E_{\pm}$ defined by $\gm( t^{k \pm}_{\g}, t^{\pm}_{l}) = \delta^{k}_{l}$. The advantage of this basis is that it also satisfies the relation $\< t^{k \pm}_{\gm}, t^{\pm}_{l} \>_{\d} = \pm \delta^{k}_{l}$. There is nothing interesting on the evaluation of $\RS_{\cD^{0}}^{+}$ and we take the liberty to omit the routine manipulations leading to the following expression. We use the similar tricks as to deal with the $(+++)$ and $(---)$ components as in the calculation of $\RS_{\cD^{0}}$. The mixed ones are treated using the consequences of the compatibility of $\cD^{0}$ with $\E_{+}$ and the torsion-free property (\ref{eq_kpluscyclic}) in the form 
\begin{equation} 
\begin{split}
\fk(x_{+},y_{-},z_{-}) = & \ - \chi(x_{+},y_{-},z_{-}), \\
\fk(x_{-},y_{+},z_{+}) = & \ - \chi(x_{-},y_{+},z_{+}), 
\end{split}
\end{equation}
for all $x,y,z \in \d$. 
The resulting expression for $\RS^{+}_{\cD^{0}}$ is 
\begin{equation}
\begin{split}
\RS_{\cD^{0}}^{+} = & \ \frac{1}{6} \chi(t_{i}^{+},t_{j}^{+},t_{k}^{+}) \chi( t^{i+}_{\gm}, t^{j+}_{\gm}, t^{k+}_{\gm})\,+ \\
& + \frac{1}{6} \chi(t_{i}^{-},t_{j}^{-},t_{k}^{-}) \chi( t^{i-}_{\gm}, t^{j-}_{\gm}, t^{k-}_{\gm})\,- \\
& - \frac{1}{2} \chi(t_{i}^{-}, t_{j}^{+}, t_{k}^{+}) \chi(t^{i-}_{\gm}, t^{j+}_{\gm}, t^{k+}_{\gm})\,- \\
& - \frac{1}{2} \chi(t_{i}^{+}, t_{j}^{-}, t_{k}^{-}) \chi(t^{i+}_{\gm}, t^{j-}_{\gm}, t^{k-}_{\gm}).
\end{split}
\end{equation}
First, note that $\RS^{+}_{\cD^{0}}$ indeed does not depend on the choice of the divergence-free connection $\cD^{0}$. There are now several ways how to rewrite $\RS^{+}_{\cD^{0}}$, none of them particularly elegant. For example, if we define $\<\chi,\chi\>_{\gm} = \frac{1}{6} \chi( t_{\mu}, t_{\nu}, t_{\lambda}) \chi( t^{\mu}_{\gm}, t^{\nu}_{\gm}, t^{\lambda}_{\gm})$, one has 
\begin{equation}
\begin{split}
\<\chi,\chi\>_{\gm} = & \ \frac{1}{6} \chi(t_{i}^{+},t_{j}^{+},t_{k}^{+}) \chi( t^{i+}_{\gm}, t^{j+}_{\gm}, t^{k+}_{\gm}) \,+\\
& + \frac{1}{6} \chi(t_{i}^{-},t_{j}^{-},t_{k}^{-}) \chi( t^{i-}_{\gm}, t^{j-}_{\gm}, t^{k-}_{\gm})\,+ \\
& + \frac{1}{2} \chi(t_{i}^{-}, t_{j}^{+}, t_{k}^{+}) \chi(t^{i-}_{\gm}, t^{j+}_{\gm}, t^{k+}_{\gm})\,+ \\
& + \frac{1}{2} \chi(t_{i}^{+}, t_{j}^{-}, t_{k}^{-}) \chi(t^{i+}_{\gm}, t^{j-}_{\gm}, t^{k-}_{\gm}).
\end{split}
\end{equation}
We see that the last two terms come with the different sign. One thus finds
\begin{equation}
\begin{split}
\RS^{+}_{\cD^{0}} = & \ \<\chi,\chi\>_{\gm} - \chi(t_{i}^{-},t_{j}^{+}, t_{k}^{+}) \chi( t^{i-}_{\gm}, t^{j+}_{\gm}, t^{k+}_{\gm})\,- \\
& - \chi(t_{i}^{+}, t_{j}^{-}, t_{k}^{-}) \chi(t^{i+}_{\gm}, t^{j-}_{\gm}, t^{k-}_{\gm}).
\end{split}
\end{equation}
This can be mildly simplified further, we can write
\begin{equation}
\RS^{+}_{\cD^{0}} = \<\chi,\chi\>_{\gm} - \gm( [t_{i}^{-}, t_{j}^{+}]_{\d}, [t^{i-}_{\gm}, t^{j+}_{\gm}]_{\d}). 
\end{equation}
We will use this expression as it contains no unnecessary prefactors. Finally, let us turn our attention to the tensorial part of the equations. Again using the split basis for the summation in (\ref{eq_Ric0eqfinal}), for all $x,y \in \d$, one immediately obtains 
\begin{equation}
\Ric_{\cD^{0}}(x_{+},y_{-}) = -\chi( t_{i}^{-}, t_{j}^{+}, x_{+}) \chi(t^{i-}_{\gm}, t^{j+}_{\gm}, y_{-}). 
\end{equation}
See that the right-hand side does not depend on the choice of divergence-free $\cD^{0}$ at all, as we claimed before. Let us summarize this subsection with the proposition. 
\begin{tvrz}
The system of algebraic equations of Theorem \ref{thm_main2} can be rewritten as 
\begin{subequations}
\begin{align}
\label{eq_eom1} 0 = & \ \RS_{\cD^{0}}^{+} =  \<\chi,\chi\>_{\gm} - \gm( [t_{i}^{-},t_{j}^{+}]_{\d}, [t^{i-}_{\gm}, t^{j+}_{\gm}]_{\d}), \\
\label{eq_eom2} 0 = & \ \Ric_{\cD^{0}}(x_{+},y_{-}) = -\chi(t_{i}^{-},t_{j}^{+},x_{+}) \chi( t^{i-}_{\gm}, t^{j+}_{\gm}, y_{-}), 
\end{align}
\end{subequations}
where $\gm \in S^{2}(\d^{\ast})$ is the positive-definite metric on $\d$ induced by the generalized metric $\E_{+} \subset \d$, $\chi \in \Lambda^{3}(\d^{\ast})$ is the Cartan $3$-form $\chi(x,y,z) = \<[x,y]_{\d},z \>_{\d}$ and $x_{\pm}$ are projections onto $\E_{\pm}$. We use the bases $(t_{i}^{\pm})_{i=1}^{\dim(\g)}$ and $(t^{i \pm}_{\gm})_{i=1}^{\dim(\g)}$ of $\E_{\pm}$ defined above. There also holds an auxiliary equation
\begin{equation}
0 = \< \chi, \chi \>_{\d}. 
\end{equation}
The numbers $\<\chi,\chi\>_{\gm}$ and $\<\chi,\chi\>_{\d}$ are defined via the standard (pseudo)scalar products on $\Lambda^{3}(\d^{\ast})$ induced by the metrics $\gm$ and $\d$. 
\end{tvrz}
\subsection{Equations in the convenient splitting}
So far, we have not used the richer algebraic structure of the Manin pair $(\d,\g)$. To do so, we will now conveniently choose an isotropic splitting of the sequence (\ref{eq_ManinSES}). Let $E_{0} \in \Hom(\g^{\ast},\g)$ be the invertible map associated to the given isotropic splitting $\tj: \g^{\ast} \rightarrow \d$ as in (\ref{eq_defE0map}). We can uniquely decompose it into its symmetric and skew-symmetric part as $E_{0} = g_{0}^{-1} + \theta_{0}$, where $g_{0} \in S^{2}(\g^{\ast})$ is a positive-definite metric on $\g$ and $\theta_{0} \in \Lambda^{2}(\g)$. It follows that $\gm \in S^{2}(\d^{\ast})$ can be, with respect to the isomorphism $\d \cong \g \oplus \g^{\ast}$, induced by $\tj$ written as a product of block matrices:
\begin{equation} \label{eq_fgtj}
\begin{split}
\gm_{\tj} = & \  \bma{1}{0}{\theta_{0}}{1} \bma{g_{0}}{0}{0}{g_{0}^{-1}} \bma{1}{-\theta_{0}}{0}{1} \\
= & \ \bma{g_{0}}{-g_{0} \theta_{0}}{\theta_{0} g_{0}}{g_{0}^{-1} - \theta_{0} g_{0} \theta_{0}}.
\end{split}
\end{equation}
Note that $g_{0}$ is precisely the pullback of $\gm$ by the inclusion $\ti: \g \rightarrow \d$. Let $\tj'$ be a new isotropic splitting obtained by choosing $\theta = \theta_{0}$ in (\ref{eq_splitchange}). As already noted in (\ref{eq_splitchangedilaton}), this results in $E'_{0} = g_{0}^{-1}$, in particular
\begin{equation}
\gm_{\tj'} = \bma{g_{0}}{0}{0}{g_{0}^{-1}}. 
\end{equation}
This shows that we can always choose a unique isotropic splitting $\tj$ so that $\gm_{\tj}$ is block diagonal. We assume that this is the case in the remainder of this subsection.

As a warm up exercise, let us prove that if $\g \subseteq \d$ is unimodular, one has $\<\chi,\chi\>_{\d} = 0$. As we fixed our splitting, we may identify $\d$ with $\g \oplus \g^{\ast}$ equipped with the bracket $(\ref{eq_doublebracket})$ and the canonical pairing $\<\cdot,\cdot\>_{\d}$. It is then natural to consider a basis of $\g \oplus \g^{\ast}$ in the form
\begin{equation}
(t_{\mu})_{\mu=1}^{\dim(\d)} = (t_{k},0)_{k=1}^{\dim(\g)} \cup (0,t^{k})_{k=1}^{\dim(\g)},
\end{equation}
where $(t_{k})_{k=1}^{\dim(\g)}$ is any basis of $\g$ and $(t^{k})_{k=1}^{\dim(\g)}$ is the corresponding dual one. It is then straightforward to arrive to the expression 
\begin{equation}
\< \chi, \chi \>_{\d} = \< [t_{i},t_{j}]_{\g}, [t^{i},t^{j}]_{\g^{\ast}} \>.
\end{equation}
As noted below (\ref{eq_doublebracket}), there is a compatibility condition between the bracket $[\cdot,\cdot]_{\g}$ and the cobracket $\delta$. Namely,
\begin{equation}
\delta([x,y]_{\g}) = \ad^{(2)}_{x}( \delta(y)) - \ad^{(2)}_{y}( \delta(x)), 
\end{equation}
for all $x,y \in \g$, where $\delta(x)(\xi,\eta) = \<[\xi,\eta]_{\g^{\ast}},x\> $ and $\ad^{(2)}$ is the adjoint representation of $\g$ on $\Lambda^{2}(\g)$. By evaluating this equation on $\xi,\eta \in \g^{\ast}$, one finds
\begin{equation} \label{eq_duality}
\begin{split}
\<[x,y]_{\g}, [\xi,\eta]_{\g^{\ast}} \> = & \ \< \ad^{\ast}_{\xi}(y), \ad^{\ast}_{x}(\eta)\>\,+ \\
& + \< \ad^{\ast}_{y}(\xi), \ad^{\ast}_{\eta}(x) \>\,- \\
& - \< \ad^{\ast}_{\eta}(y), \ad^{\ast}_{x}(\xi)\>\,- \\
& - \< \ad^{\ast}_{y}(\eta), \ad^{\ast}_{\xi}(x) \>.
\end{split}
\end{equation}
Choosing $(x,y,\xi,\eta) = (t_{i},t_{j},t^{i},t^{j})$ then gives 
\begin{equation}
\< \chi,\chi\>_{\d} = 2 \< \falpha, \fa \> = \< (\fa,\falpha), (\fa, \falpha) \>_{\d},
\end{equation}
where $\falpha(x) = \Tr(\ad_{x})$ and $\xi(\fa) = \Tr(\ad_{\xi})$ for all $x \in \g$ and $\xi \in \g^{\ast}$. In particular, if $\g \subseteq \d$ is unimodular, we indeed obtain $\<\chi,\chi\>_{\d} = 0$. 

Next, we have to analyze the right-hand side of (\ref{eq_eom1}). Let us start with $\<\chi,\chi\>_{\gm}$. We have
\begin{equation}
\<\chi,\chi\>_{\gm} = \frac{1}{6} \gm( [t_{\mu},t_{\nu}]_{\d}, [t^{\mu}_{\gm}, t^{\nu}_{\gm}]_{\d}). 
\end{equation} 
In the remainder of this subsection, let us use the notation $e_{k} = g_{0}(t_{k})$ and $e^{k} = g_{0}^{-1}(t^{k})$. By plugging into (\ref{eq_doublebracket}) and using the fact that $\gm$ is block diagonal, one finds
\begin{equation}
\begin{split}
\<\chi,\chi\>_{\gm} = & \ \frac{1}{2} g_{0}( [t_{i},t_{j}]_{\g}, [e^{i},e^{j}]_{\g})\,+ \\
& + \frac{1}{2} g_{0}^{-1}( [t^{i},t^{j}]_{\g^{\ast}}, [e_{i},e_{j}]_{\g^{\ast}}) + \< \mu, \mu\>_{g_{0}}.
\end{split}
\end{equation}
To treat the second term on the right-hand side of (\ref{eq_eom1}), choose the basis vectors as 
\begin{equation}
\begin{split}
t_{k}^{+} = & \ (t_{k},e_{k}), \; \; t_{k}^{-} = (t_{k},-e_{k}), \\
t^{k+}_{\gm} = & \  \frac{1}{2} (e^{k},t^{k}), \; \; t^{k-}_{\gm} = \frac{1}{2}(e^{k},-t^{k}).
\end{split}
\end{equation}
Using this basis and a series of straightforward manipulations, we arrive to the expression
\begin{equation}
\begin{split}
\gm([t_{i}^{-},t_{j}^{+}]_{\d} &, [t^{i-}_{\gm}, t^{j+}_{\gm}]_{\d}) =\\
= & \ \frac{3}{4} g_{0}([t_{i},t_{j}]_{\g}, [e^{i},e^{j}]_{\g})\,+ \\
& + \frac{3}{4} g_{0}^{-1}([t^{i},t^{j}]_{\g^{\ast}}, [e_{i},e_{j}]_{\g^{\ast}})\,+ \\
& + \frac{3}{2} \<\mu,\mu\>_{g_{0}} + \frac{1}{2} \Tr_{g_{0}}(c_{\g}) + \frac{1}{2} \Tr_{g_{0}}(c_{\g^{\ast}})\,- \\
& - \frac{1}{2} \mu(t^{i},t^{j}, g_{0}([t_{i},t_{j}]_{\g})),
\end{split}
\end{equation}
where $c_{\g}$ is the Killing form of $\g$ and $c_{\g^{\ast}}$ is the symmetric bilinear form on $\g^{\ast}$ obtained by the same formula from the bracket $[\cdot,\cdot]_{\g^{\ast}}$ (remember that this is not a Lie algebra in general). Plugging the two partial results into the right-hand side of (\ref{eq_eom1}) leads to the formula
\begin{equation}
\begin{split}
\RS_{\cD^{0}}^{+} = & \ - \frac{1}{4} g_{0}([t_{i},t_{j}]_{\g}, [e^{i},e^{j}]_{\g})\,- \\
& - \frac{1}{4} g_{0}^{-1}( [t^{i},t^{j}]_{\g^{\ast}}, [e_{i},e_{j}]_{\g^{\ast}})\,- \\
& - \frac{1}{2}\<\mu,\mu\>_{g_{0}} - \frac{1}{2} \Tr_{g_{0}}(c_{\g}) - \frac{1}{2} \Tr_{g_{0}}(c_{\g^{\ast}})\,+ \\
& + \frac{1}{2} \mu(t^{i},t^{j}, g_{0}([t_{i},t_{j}]_{\g})). 
\end{split}
\end{equation}
Finally, to analyze the equation (\ref{eq_eom2}), we will introduce a tensor $\Ric^{0}$ on $\g$ via the formula
\begin{equation}
\Ric^{0}(x,y) = \Ric_{\cD^{0}}((x,g_{0}(x)), (y,-g_{0}(y)),
\end{equation}
for all $x,y \in \g$. The condition (\ref{eq_eom2}) is then equivalent to $\Ric^{0} = 0$. This tensor can be uniquely decomposed as $\Ric^{0} = \Ric^{0}_{s} + \Ric^{0}_{a}$ where $\Ric^{0}_{s} \in S^{2}(\g^{\ast})$ and $\Ric^{0}_{a} \in \Lambda^{2}(\g^{\ast})$. The symmetric part is given by the formula
\begin{equation}
\begin{split}
\Ric^{0}_{s}(x,y) = & \ \frac{1}{4} g_{0}([t_{i},t_{j}]_{\g},x) \cdot g_{0}([e^{i},e^{j}]_{\g},y)\,- \\
& - \frac{1}{4} \< [t^{i},t^{j}]_{\g^{\ast}},x\> \cdot \< [e_{i},e_{j}]_{\g^{\ast}},y \>\,- \\
& - \frac{1}{4} \mu(t^{i},t^{j},g_{0}(x)) \cdot g_{0}([t_{i},t_{j}]_{\g},y)\,- \\
& - \frac{1}{4} \mu(t^{i},t^{j},g_{0}(y)) \cdot g_{0}([t_{i},t_{j}]_{\g}, x)\,+ \\
& + \frac{1}{2} g_{0}^{-1}( [g_{0}(x),t^{i}]_{\g^{\ast}}, [g_{0}(y),e_{i}]_{\g^{\ast}})\,- \\
& - \frac{1}{2} g_{0}([x,t_{i}]_{\g}, [y,e^{i}]_{\g})\,+ \\
& + \frac{1}{2} c_{\g^{\ast}}( g_{0}(x), g_{0}(y)) - \frac{1}{2} c_{\g}(x,y)\,+ \\
& + \frac{1}{4} \mu(t^{i},t^{j},g_{0}(x)) \cdot \mu(e_{i},e_{j},g_{0}(y)).
\end{split}
\end{equation}
For the skew-symmetric part, one finds
\begin{equation}
\begin{split}
\Ric^{0}_{a}(x,y) = & \ \frac{1}{4} g_{0}([t_{i},t_{j}]_{\g}, x ) \cdot \< [t^{i},t^{j}]_{\g^{\ast}}, y \>\,- \\
& - \frac{1}{4} g_{0}([t_{i},t_{j}]_{\g}, y ) \cdot \< [t^{i},t^{j}]_{\g^{\ast}}, x \>\,+ \\
& + \frac{1}{4} \< [t^{i},t^{j}]_{\g^{\ast}}, x \> \cdot \mu(e_{i},e_{j},g_{0}(y))\,- \\
& - \frac{1}{4} \< [t^{i},t^{j}]_{\g^{\ast}}, y \> \cdot \mu(e_{i},e_{j},g_{0}(x))\,+ \\
& + \frac{1}{2} \< [ g_{0}(x),t^{i}]_{\g^{\ast}}, [y,t_{i}]_{\g} \>\,- \\
& - \frac{1}{2} \< [g_{0}(y),t^{i}]_{\g^{\ast}}, [x,t_{i}]_{\g} \>\,+ \\
& + \frac{1}{2} \< e^{i}, [g_{0}(x), g_{0}( [y,t_{i}]_{\g})]_{\g^{\ast}} \>\,- \\
& - \frac{1}{2} \< e^{i}, [g_{0}(y), g_{0}( [x,t_{i}]_{\g})]_{\g^{\ast}} \>.
\end{split}
\end{equation}
Some of the terms can now be slightly rearranged using the duality condition (\ref{eq_duality}). One has
\begin{equation}
\begin{split}
g_{0}([t_{i},t_{j}]_{\g},x) \cdot \<[t^{i},t^{j}]_{\g^{\ast}}, y \> = & \ \<y, [\falpha, g_{0}(x)]_{\g^{\ast}} \>\,+ \\
& + g_{0}(x, [\fa,y]_{\g}). 
\end{split}
\end{equation}
In particular, for unimodular $\g$ the first term disappears. We conclude with the formula
\begin{equation}
\begin{split}
\Ric^{0}_{a}(x,y) = & \ \frac{1}{4} g_{0}(x, [\fa,y]_{\g}) - \frac{1}{4} g_{0}(y, [\fa,x]_{\g})\,+ \\
& + \frac{1}{4} \< [t^{i},t^{j}]_{\g^{\ast}}, x \> \cdot \mu(e_{i},e_{j},g_{0}(y))\,- \\
& - \frac{1}{4} \< [t^{i},t^{j}]_{\g^{\ast}}, y \> \cdot \mu(e_{i},e_{j},g_{0}(x))\,+ \\
& + \frac{1}{2} \< [ g_{0}(x),t^{i}]_{\g^{\ast}}, [y,t_{i}]_{\g} \>\,- \\
& - \frac{1}{2} \< [g_{0}(y),t^{i}]_{\g^{\ast}}, [x,t_{i}]_{\g} \>\,+ \\
& + \frac{1}{2} \< e^{i}, [g_{0}(x), g_{0}( [y,t_{i}]_{\g})]_{\g^{\ast}} \>\,- \\
& - \frac{1}{2} \< e^{i}, [g_{0}(y), g_{0}( [x,t_{i}]_{\g})]_{\g^{\ast}} \>.
\end{split}
\end{equation}
Next, note that we can take the trace of the symmetric tensor $\Ric_{s}^{0}$ using the metric $g_{0}$. We thus obtain another scalar (which depends on the splitting) $\RS^{0}_{\cD^{0}} = \Tr_{g_{0}}( \Ric^{0}_{s})$. Explicitly, one finds
\begin{equation}
\begin{split}
\RS^{0}_{\cD^{0}} = & \ - \frac{1}{4} g_{0}([t_{i},t_{j}]_{\g}, [e^{i},e^{j}]_{\g})\,+ \\
& + \frac{1}{4} g_{0}^{-1}([ t^{i},t^{j}]_{\g^{\ast}}, [e_{i},e_{j}]_{\g^{\ast}})\,+ \\
& + \frac{3}{2}\<\mu,\mu\>_{g_{0}} - \frac{1}{2} \Tr_{g_{0}}(c_{\g}) + \frac{1}{2} \Tr_{g_{0}}(c_{\g^{\ast}})\,- \\
& - \frac{1}{2} \mu(t^{i},t^{j}, g_{0}([t_{i},t_{j}]_{\g})).
\end{split}
\end{equation}
This shows that the scalar equation $\RS^{+}_{\cD^{0}} = 0$ is not independent of the tensorial equation $\Ric^{0} = 0$. This is not very surprising as exactly the same thing happens for (super)gravity equations of motion. In fact, by adding and subtracting the two scalars, we obtain
\begin{subequations}
\begin{align}
\RS_{\cD^{0}}^{+} + \RS_{\cD^{0}}^{0} = & \ - \frac{1}{2} g_{0}([t_{i},t_{j}]_{\g}, [e^{i},e^{j}]_{\g})\,- \\
&  - \Tr_{g_{0}}(c_{\g}) + \<\mu,\mu\>_{g_{0}}, \nonumber \\ 
\RS_{\cD^{0}}^{+} - \RS_{\cD^{0}}^{0} = & \ - \frac{1}{2} g_{0}^{-1}( [t^{i},t^{j}]_{\g^{\ast}}, [e_{i},e_{j}]_{\g^{\ast}})\,- \\
 & - \Tr_{g_{0}}(c_{\g^{\ast}})  - 2 \<\mu,\mu\>_{g_{0}}\,+ \\
 & + \mu(t^{i},t^{j}, g_{0}([t_{i},t_{j}]_{\g})). \nonumber
\end{align}
\end{subequations}
We will make use of these two scalar equations in the examples.
\subsection{An arbitrary splitting}
The derivation of equations in the previous subsection has one serious defect. We have calculated everything in the special isotropic splitting  making the induced metric on $\g \oplus \g^{\ast}$ block diagonal. However, we do not know $\E_{+}$ in advance to find it. Fortunately, there is an easy way to make it right. Hence assume that $\tj: \g^{\ast} \rightarrow \d$ is an \textit{arbitrary but fixed} isotropic splitting of (\ref{eq_ManinSES}). 

As discussed in previous subsection, the sought for generalized metric $\E_{+}$ can be uniquely described by a pair $(g_{0}, \theta_{0})$ such that the induced fiber-wise metric $\gm_{\tj}$ on $\g \oplus \g^{\ast}$ takes the form (\ref{eq_fgtj}). 
Let $(\g,\delta,\mu)$ be the Lie quasi-bialgebra induced by the choice of the splitting $\tj$. Our goal is to write the equations for $\E_{+}$ in terms of $(\g,\delta,\mu)$ with unknown $(g_{0},\theta_{0})$.

Now, let $\tj': \g^{\ast} \rightarrow \d$ be the isotropic splitting defined by (\ref{eq_splitchange}) where we choose $\theta = \theta_{0}$. This is the (sought for) splitting where $\gm_{\tj'}$ is block diagonal. Let $(\g,\delta', \mu')$ be the Lie quasi-bialgebra induced by $\tj'$. In the previous subsection, we have found the equations of motion in terms of $(\g, \delta', \mu')$ for unknown $g_{0}$. The second variable $\theta_{0}$ is thus introduced only through the induced objects $\delta'$ and $\mu'$, as we will now demonstrate. First, we find
\begin{equation} \label{eq_deltachangeofsplitting}
\begin{split}
\delta'(x)(\xi,\eta) = & \ \< [\tj'(\xi), \tj'(\eta)]_{\d}, \io(x) \>_{\d} \\
= & \ (\delta(x)+ \ad^{(2)}_{x}(\theta_{0}))(\xi,\eta),
\end{split}
\end{equation}
where $\ad^{(2)}$ denotes the adjoint representation of $\g$ on $\Lambda^{2}(\g)$. Recall that $\delta$ is a $1$-cocycle in the Chevalley--Eilenberg complex $\frc^{\bullet}(\g, \Lambda^{2}(\g))$ with the coboundary operator $\Delta$. Then $\delta' = \delta - \Delta(\theta_{0})$, where $\theta_{0} \in \frc^{0}(\g,\Lambda^{2}(\g))$ is viewed as a $0$-cochain. On the level of brackets, we find 
\begin{equation}
[\xi,\eta]'_{\g^{\ast}} = [\xi,\eta]_{\g^{\ast}} + \ad^{\ast}_{\theta_{0}(\xi)}(\eta) - \ad^{\ast}_{\theta_{0}(\eta)}(\xi). 
\end{equation}
The expression regarding the $3$-vector $\mu'$ can be, for all $\xi,\eta,\zeta \in \g^{\ast}$, written as 
\begin{equation} \begin{split}
\mu'(\xi,\eta,\zeta) = & \ \mu(\xi,\eta,\zeta)  + \{ \theta_{0}([\xi,\eta]_{\g^{\ast}}, \zeta)\,+ \\
& + \< [\theta_{0}(\xi), \theta_{0}(\eta)]_{\g}, \zeta \> +  \cyc(\xi,\eta,\zeta) \} \\
= & \ ( \mu - d_{\ast}(\theta_{0}) + \frac{1}{2} [\! [ \theta_{0}, \theta_{0} ] \! ]_{\g} )(\xi,\eta,\zeta).
\end{split}
\end{equation}
Here $d_{\ast}: \Lambda^{\bullet}(\g) \rightarrow \Lambda^{\bullet+1}(\g)$ is the "almost" differential induced by the bracket $[\cdot,\cdot]_{\g^{\ast}}$ given by
\begin{equation}\begin{split} \label{eq_dastalmostdif}
(d_{\ast}(\nu))(\xi_{1},\dots,\xi_{p}) =\\
= \sum_{i < j} (-1)^{i+j} \nu( [\xi_{i},\xi_{j}]_{\g^{\ast}}, \dots, \hat{\xi}_{i}, \dots, \hat{\xi}_{j}, \dots, \xi_{p}),
\end{split}
\end{equation}
for any $\nu \in \Lambda^{p}(\g)$ and $\xi_{1},\dots,\xi_{p} \in \g^{\ast}$. Note that in general $d_{\ast}^{2} \neq 0$. Finally, $[\![\theta_{0},\theta_{0}]\!]_{\g} \in \Lambda^{3}(\g)$ is the algebraic Schouten-Nijenhuis bracket defined for all $\xi,\eta,\zeta \in \g^{\ast}$ as
\begin{equation} \label{eq_aSNbracket}
\begin{split}
[\![\theta_{0},\theta_{0}]\!]_{\g}(\xi,\eta,\zeta) = & \ 2 \< [ \theta_{0}(\xi), \theta_{0}(\eta)]_{\g}, \zeta \>\,+ \\
& + \cyc(\xi,\eta,\zeta). 
\end{split}
\end{equation}

Let us conclude this section with the final statement.

\begin{theorem}[\textbf{The algebraic equations of motion}] \label{thm_EOMalgebraic}
Let $(\d,\g)$ be any Manin pair with unimodular $\g$ and let $\tj: \g^{\ast} \rightarrow \d$ be any isotropic splitting of the sequence (\ref{eq_ManinSES}). Let $(\g,\delta,\mu)$ be the corresponding Lie quasi-bialgebra defined by (\ref{eq_defdeltamu}). Let $\E_{+} \subseteq \d$ be a generalized metric parametrized by $g_{0} \in S^{2}(\g^{\ast})$ and $\theta_{0} \in \Lambda^{2}(\g)$, such that 
\begin{equation}
\E_{+} = \{ \ti( (g_{0}^{-1} + \theta_{0})(\xi)) + \tj(\xi) \; | \; \xi \in \g^{\ast} \}. 
\end{equation}
Let $\delta' = \delta - \Delta(\theta_{0})$ and $\mu' = \mu - d_{\ast}(\theta_{0}) + \frac{1}{2} [\![ \theta_{0}, \theta_{0} ]\!]_{\g}$, where the notation is explained in the preceding paragraphs\footnote{c.f. text between equations (\ref{eq_deltachangeofsplitting}) and (\ref{eq_aSNbracket}).}. Let $[\cdot,\cdot]'_{\g^{\ast}}$ be the $\R$-bilinear skew-symmetric bracket on $\g^{\ast}$ corresponding to $\delta'$. Let $\cD^{0} \in \LC(\d,\E_{+})$ be an arbitrary divergence-free Levi-Civita connection. 

Then the generalized scalar curvature $\RS^{+}_{\cD^{0}}$ of $\cD^{0}$ with respect to $\E_{+}$ can be written as 
\begin{equation} \label{eq_EOMscalar1}
\begin{split}
\RS^{+}_{\cD^{0}} = & \ - \frac{1}{4} g_{0}([t_{i},t_{j}]_{\g}, [e^{i},e^{j}]_{\g})\,- \\
& - \frac{1}{4} g_{0}^{-1}([t^{i},t^{j}]'_{\g^{\ast}}, [e_{i},e_{j}]'_{\g^{\ast}})\,- \\
& - \frac{1}{2} \<\mu',\mu'\>_{g_{0}} - \frac{1}{2} \Tr_{g_{0}}(c_{\g}) - \frac{1}{2} \Tr_{g_{0}}(c'_{\g^{\ast}})\,+ \\
& + \frac{1}{2} \mu'(t^{i},t^{j}, g_{0}([t_{i},t_{j}]_{\g})).
\end{split}
\end{equation}
The tensorial equation $\Ric_{\cD^{0}}(\E_{+},\E_{-}) = 0 $ is equivalent to the vanishing of a tensor $\Ric^{0}$ on $\g$ whose symmetric and skew-symmetric parts are defined as
\begin{subequations}
\begin{equation} \label{eq_EOMtensor1}
\begin{split}
\Ric^{0}_{s}(x,y) = & \ \frac{1}{4} g_{0}([t_{i},t_{j}]_{\g},x) \cdot g_{0}([e^{i},e^{j}]_{\g},y)\,- \\
& - \frac{1}{4} \< [t^{i},t^{j}]'_{\g^{\ast}},x\> \cdot \< [e_{i},e_{j}]'_{\g^{\ast}},y \>\,- \\
& - \frac{1}{4} \mu'(t^{i},t^{j},g_{0}(x)) \cdot g_{0}([t_{i},t_{j}]_{\g},y)\,- \\
& - \frac{1}{4} \mu'(t^{i},t^{j},g_{0}(y)) \cdot g_{0}([t_{i},t_{j}]_{\g}, x)\,+ \\
& + \frac{1}{2} g_{0}^{-1}( [g_{0}(x),t^{i}]'_{\g^{\ast}}, [g_{0}(y),e_{i}]'_{\g^{\ast}})\,- \\
& - \frac{1}{2} g_{0}([x,t_{i}]_{\g}, [y,e^{i}]_{\g})\,+ \\
& + \frac{1}{2} c'_{\g^{\ast}}( g_{0}(x), g_{0}(y)) - \frac{1}{2} c_{\g}(x,y)\,+ \\
& + \frac{1}{4} \mu'(t^{i},t^{j},g_{0}(x)) \cdot \mu'(e_{i},e_{j},g_{0}(y)),
\end{split}
\end{equation}
\begin{equation} \label{eq_EOMtensor2}
\begin{split}
\Ric^{0}_{a}(x,y) = & \ \frac{1}{4} g_{0}(x, [\fa',y]_{\g}) - \frac{1}{4} g_{0}(y, [\fa',x]_{\g})\,+ \\
& + \frac{1}{4} \< [t^{i},t^{j}]'_{\g^{\ast}}, x \> \cdot \mu'(e_{i},e_{j},g_{0}(y))\,- \\
& - \frac{1}{4} \< [t^{i},t^{j}]'_{\g^{\ast}}, y \> \cdot \mu'(e_{i},e_{j},g_{0}(x))\,+ \\
& + \frac{1}{2} \< [ g_{0}(x),t^{i}]'_{\g^{\ast}}, [y,t_{i}]_{\g} \>\,- \\
& - \frac{1}{2} \< [g_{0}(y),t^{i}]'_{\g^{\ast}}, [x,t_{i}]_{\g} \>\,+ \\
& + \frac{1}{2} \< e^{i}, [g_{0}(x), g_{0}( [y,t_{i}]_{\g})]'_{\g^{\ast}} \>\,- \\
& - \frac{1}{2} \< e^{i}, [g_{0}(y), g_{0}( [x,t_{i}]_{\g})]'_{\g^{\ast}} \>.
\end{split}
\end{equation}
\end{subequations}
Here $(t_{i})_{i=1}^{\dim(\g)}$ and $(t^{i})_{i=1}^{\dim(\g)}$ denote mutually dual bases of $\g$ and $\g^{\ast}$, respectively. We write $e_{i} = g_{0}(t_{i})$, $e^{i} = g_{0}^{-1}(e^{i})$. By $c_{\g}$ we denote the Killing form of $\g$ and $c'_{\g^{\ast}}$ is the symmetric bilinear form on $\g^{\ast}$ defined using the same formula and bracket $[\cdot,\cdot]'_{\g^{\ast}}$. Finally, $\xi(\fa') = \Tr(\ad'_{\xi})$. 

If the tensorial equation holds together with the scalar condition $\RS_{\cD^{0}}^{+} = 0 $, the following two scalar equations must hold as well:
\begin{subequations}
\begin{align}
\label{eq_EOMscalar2} 0 = & \ - \frac{1}{2} g_{0}([t_{i},t_{j}]_{\g}, [e^{i},e^{j}]_{\g}) - \Tr_{g_{0}}(c_{\g}) + \<\mu',\mu'\>_{g_{0}}, \\
\label{eq_EOMscalar3} 0 = & \ - \frac{1}{2} g_{0}^{-1}( [t^{i},t^{j}]'_{\g^{\ast}}, [e_{i},e_{j}]'_{\g^{\ast}})  - \Tr_{g_{0}}(c'_{\g^{\ast}})\,-  \\
& \ - 2 \<\mu',\mu'\>_{g_{0}} + \mu'(t^{i},t^{j}, g_{0}([t_{i},t_{j}]_{\g})). \nonumber
\end{align}
\end{subequations}
\end{theorem}
\begin{proof}
As already discussed, it suffices to take the results of the previous subsection and everywhere replace $[\cdot,\cdot]_{\g^{\ast}}$ and $\mu$ by $[\cdot,\cdot]'_{\g^{\ast}}$ and $\mu'$, respectively. 
\end{proof}
\end{appendices}

\bibliographystyle{prop2015}
\bibliography{allbibtex}

\end{document}